\documentclass[12pt]{article}

\usepackage{hyperref}

\usepackage{amsfonts}

\usepackage{amscd}
\usepackage{amsthm}
\usepackage{indentfirst}

\usepackage{amssymb, amsmath, amsthm, amsgen, amstext, amsbsy, amsopn}
\usepackage{epsfig}

\usepackage{color}
\usepackage{enumerate}
\usepackage{enumitem}

\newtheorem{thm}{Theorem}[section]
\newtheorem{lemma}[thm]{Lemma}
\newtheorem{prop}[thm]{Proposition}

\newtheorem{defin}[thm]{Definition}
\newtheorem{rem}[thm]{Remark}
\newtheorem{exam}[thm]{Example}

\newtheorem{question}[thm]{Question}

\newcommand{\R}{{\mathbb{R}}}
\newcommand{\EE}{{\mathbb{E}}}

\newcommand{\Z}{{\mathbb{Z}}}
\newcommand{\N}{{\mathbb{N}}}
\newcommand{\C}{{\mathbb{C}}}

\newcommand{\cC}{{\mathcal{C}}}

\newcommand{\cL}{{\mathcal{L}}}

\newcommand{\cN}{{\mathcal{N}}}

\newcommand{\cU}{{\mathcal{U}}}

\def\id{{1\hskip-2.5pt{\rm l}}}

\newcommand{\be}{{\beta}}
\newcommand{\al}{{\alpha}}
\newcommand{\la}{{\lambda}}
\newcommand{\De}{{\Delta}}
\newcommand{\si}{{\sigma}}
\newcommand{\hb}{{\hbar}}

\newcommand{\ga}{{\gamma}}

\newcommand{\om}{{\omega}}
\newcommand{\Om}{\Omega}

\newcommand{\Ci}{{\mathcal{C}}^{\infty}}
\newcommand{\Cl}{\mathcal{C}}
\newcommand{\op}{\operatorname}
\newcommand{\con}{\overline}
\newcommand{\bigo}{\mathcal{O}}
\newcommand{\Hilb}{\mathcal{H}}
\newcommand{\bigoinf}{\mathcal{O}_{\infty}}  

\newcommand{\B}{\mathcal{B}}
\newcommand{\Cb}{\Cl_{\op{b}}}

\newcommand{\Lap}{\Delta}     

\begin{document}

\title{Sharp correspondence principle and quantum measurements}

\renewcommand{\thefootnote}{\alph{footnote}}

\author{\textsc Laurent Charles  and Leonid Polterovich$^{a}$ }

\footnotetext[1]{Partially supported by the Israel Science Foundation grant 178/13 and the European Research Council Advanced grant 338809.}

\date{\today}

\maketitle

\centerline{{\it To Yuri Dmitrievich Burago on the occasion of his 80th birthday}}

\begin{abstract} We prove sharp remainder bounds for the Berezin-Toeplitz quantization
and present applications to semiclassical quantum measurements.
\end{abstract}

\tableofcontents

\section{Introduction and main results}

\subsection{An outlook}  The subject of this paper is quantization, a formalism behind the quantum-classical correspondence, a fundamental principle stating that quantum mechanics contains the classical one as a limiting case when the Planck constant $\hbar$ tends to $0$. Mathematically, the correspondence is given by a linear map between smooth functions on a symplectic manifold and Hermitian operators on an $\hbar$-depending complex Hilbert space. It is assumed that, up to error terms (a.k.a. {\it remainders}) which are small with $\hbar$, some basic operations on functions correspond to their counterparts on operators. For instance, the Poisson bracket of functions corresponds to (a properly rescaled)  commutator of the operators. In the present paper we study the size of the remainders focusing on the following facets of this problem.

First, {\it given a quantization, find explicit upper bounds for the remainders.} In this direction,
we obtain such bounds for several basic quantization schemes, including the standard Berezin-Toeplitz quantization of closed K\"{a}hler manifolds, in terms of the $\Cl^k$-norms on functions with $k \leq 3$
(see Sections \ref{subsec-sharp-remainders} and \ref{bargmann} below).
The motivation comes from semiclassical quantum mechanics: having a good control on the remainders, one can zoom into small regions of the phase space up to the quantum scale $\sim \sqrt{\hbar}$, the smallest scale allowed by the uncertainty principle. As an illustration, in Section \ref{subsec-joint}
we present applications to noise production in semiclassical quantum measurements.

Second, according to the classical no-go theorems, no ideal quantizations, i.e., the ones
with vanishing remainders, exist. This naturally leads to the following quantitative question:
{\it can find a quantization with arbitrary small remainders?} It turns out that for a meaningful class of quantizations, the answer is negative. The remainders are subject to constraints which depend only on geometry of the phase space, which is the essence of the  rigidity of remainders phenomenon discussed in Section \ref{subsec-ror}.

\subsection{Sharp remainder estimates}\label{subsec-sharp-remainders}

Let $(M^{2n},\omega)$ be a closed symplectic manifold. We assume that $(M,\omega)$ is quantizable,
i.e., the cohomology class $[\omega]/(2\pi)$ is integral. We write $\{f,g\}$ for the Poisson bracket
of smooth functions $f$ and $g$.

Fix an auxiliary Riemannian metric $\rho$ on $M$. For a function
$f \in \Cl^\infty(M)$ its $\Cl^k$-norm with respect to $\rho$ is denoted by
$|f|_{k}$. For a pair of smooth functions $f,g$ put
$|f,g|_{N} = \sum_{j=0}^N |f|_{j} \cdot |g|_{N-j}$.
We write $||f||=|f|_0 =\max |f|$ for the uniform norm, and $||f||_{L_1}$ for the $L_1$-norm of $f$
with respect to the symplectic volume $\omega^n/n!$.

We also introduce a {\it reduced} version of $|f,g|_{4}$,
\begin{equation}\label{eq-reduced}
|f,g|_{1,3} := |f|_{1} \cdot |g|_{3} +  |f|_{2} \cdot |g|_{2}+ |f|_{3} \cdot |g|_{1}
\end{equation}
which does not include fourth derivatives, and which plays an important role below.

For a finite-dimensional complex Hilbert space $\Hilb$ write $\cL(\Hilb)$ for the space of Hermitian operators on $\Hilb$. The operator norm is denoted by $||\cdot||_{op}$ and $[A,B]$
stands for the commutator $AB-BA$.

\medskip

A {\it Berezin}
quantization of $M$ is given by the following data:
\begin{itemize}
\item a subset $\Lambda \subset \R_{>0}$ having $0$ as a limit point;
\item a family $\Hilb_\hbar$ of finite-dimensional
complex Hilbert spaces, $\hbar \in \Lambda$;
\item a family $T_{\hbar} : \Cl (M, \R) \rightarrow \cL(\Hilb_\hbar)$ of positive surjective linear maps such that $T_{\hbar} (1) = \op {id}$, $\hbar \in \Lambda$.
\end{itemize}

That $T_{\hbar}$ is positive means that for any $f \in \Cl (M, \R)$,  $f \geq 0$ implies that $T_{\hbar} (f) \geq 0$. We will also assume that
there exist $\alpha,\beta,\gamma$ and $\delta >0$ such that for any $f,g \in \Ci ( M , \R)$, we have

\begin{itemize}
\item[{(P1)}] {\bf (norm correspondence)} $||f||- \alpha |f|_{2}\hbar \leq ||T_\hbar(f)||_{op} \leq ||f||$;
\item[{(P2)}] {\bf (the correspondence principle)} $$|| -\frac{i}{\hbar} \cdot [T_{\hbar}(f),T_{\hbar}(g)] - T_\hbar (\{f,g\})||_{op} \leq \beta |f,g|_{1,3} \hbar\;;$$
\item[{(P3)}] {\bf (quasi-multiplicativity)} $||T_\hbar(fg) - T_\hbar(f)T_\hbar(g)||_{op} \leq \gamma |f,g|_{2}\hbar$;
\item[{(P4)}] {\bf (trace correspondence)}  $$ \big{|} \text{trace}(T_\hbar(f)) - (2\pi\hbar)^{-n}\int_M f \frac{\omega^n}{n!}\big{|}  \leq \delta ||f||_{L_1} \hbar^{-(n-1)}\;,$$
\end{itemize}
for all $f,g \in \Cl^\infty(M)$ and all $\hbar \in \Lambda$.

\medskip
\noindent A few elementary remarks are in order. The upper bound in (P1) is an immediate
consequence of the positivity of $T_{\hbar}$.  Furthermore,
since $fg=gf$, it follows from (P3) that $[T_{\hbar}(f),T_{\hbar}(g)] = \bigo(\hbar)$. Property
(P2) can be considered as a refinement of this formula. Note also that substituting $f=1$ into (P4) we get that the dimension of the space $\Hilb_\hbar$ tends to $\infty$ as $\hbar \to 0$.


\medskip
\noindent
\begin{thm}\label{thm-main} Every quantizable symplectic manifold admits a Berezin quantization
satisfying (P1)-(P4).
\end{thm}

\medskip
\noindent
The novelty here is a fine structure of the  remainders in (P1)-(P3).
In particular, for quantizable K\"{a}hler manifolds, the standard Berezin-Toeplitz quantization
satisfies (P1)-(P4).


Interestingly enough, for fixed $\omega$ and $\rho$, the coefficients $\alpha$, $\beta$ and $\gamma$ are subject to constraints which manifest optimality of inequalities (P1)-(P3). We discuss them in the next section. As a counterpoint, for certain quantization the constant $\delta$ in (P4) can be made arbitrarily small,  see Remark \ref{rem-delta} below. Furthermore, we present applications of (P1)-(P3) to semiclassical quantum measurements.

The seminal reference on Berezin quantization is the book \cite{BoGu} by Boutet de Monvel and Guillemin. In \cite{BoMeSc}, \cite{Gu} or \cite{BU}, it was deduced from \cite{BoGu} the existence of a quantization satisfying the following version of (P1)-(P3): for any smooth functions $f,g$,
\begin{gather} \label{eq:2}
 \| T_h( f) \|_{op} = | f| + \bigo (\hb),\quad  T_\hbar(f)T_{\hbar} (g) = T_\hbar(fg) + \bigo(\hb) \\
\label{eq:3}  [T_{\hbar}(f),T_{\hbar}(g)] =  \tfrac{\hb}{i} T_\hbar (\{f,g\}) + \bigo( \hb^2)
\end{gather}
where the $\bigo$'s are in uniform norm and depend on $f$ and $g$. More recently, Barron et al \cite{BaMaMaPi} have extended (\ref{eq:2}) to functions of class $\Cl^2$ and (\ref{eq:3}) to functions of class $\Cl^4$. We will prove in this paper that actually (\ref{eq:3}) holds for functions of class $\Cl^3$ (as we shall see in Remark \ref{rem-reduced} below, the absence of fourth derivatives in
the correspondence principle (P2) has meaningful applications). Additionally, in Proposition \ref{prop:P1_kahler} below we prove a slightly stronger version of quasi-multiplicativity (P3).
More importantly, we make explicit the dependence in $f$ and $g$ of the remainders in the sense of (P1)-(P3).

We refer the reader to lecture notes \cite{LF} by Y.~Le Floch for a skillfully written exposition of our
Theorem \ref{thm-main} in the K\"{a}hler case and useful preliminaries.



\subsection{Rigidity of remainders}\label{subsec-ror}

Confronting matrix analysis with geometry of the phase space, we get the following constraints
on the remainders in (P1),(P2) and (P3) for any Berezin quantization.

\medskip
\noindent
\begin{thm} \label{thm-main-1} Let $(M,\omega)$ be closed
quantizable symplectic manifold equipped with a Riemannian metric $\rho$. There exist
positive constants $C_1,C_2,C_3$ depending on $(M,\omega,\rho)$
such that for every Berezin quantization
\begin{itemize}
\item[{(i)}] $\alpha \geq C_1$;
\item[{(ii)}] $\beta \geq C_2\alpha^{-2}$;
\item[{(iii)}] $\gamma \geq C_3$.
\end{itemize}
\end{thm}

\medskip
\noindent The proof is given in Section \ref{sec-constraints} below.

\medskip

The lower bound (ii) on the $\beta$-remainder in the correspondence principle (P2) deserves a special discussion. According to the classical no-go theorem (see, e.g., \cite{GiMo} and references therein), there is no linear map $T_\hbar: C^{\infty}(M) \to \cL(\Hilb_\hbar)$ which sends (up to a multiplicative constant)  Poisson brackets of functions to commutators of operators.
In other words, $\beta \neq 0$. As we shall see below, the proof of (ii), in addition to (P2), involves only the norm correspondence (P1) with fixed $\alpha$. Therefore, in the presence of (P1), inequality (ii) can be considered as a quantitative version of the no-go theorem.

Furthermore, the proof of (ii) involves both the lower and the upper bounds in the norm correspondence
(P1). Recall that the latter uses that the Berezin quantization is positive.  Interestingly enough, without the preservation of positivity
(ii) is not necessarily valid. For instance, for the {\it geometric quantization}
of a compact quantizable K\"{a}hler manifold in the presence of the metaplectic correction, the remainder
in (P2) is of the order $\bigo( \hb^2)$, see formula (19) in \cite{Charles-metaplectic}. Let us mention
that the preservation of positivity is crucial for our applications to quantum measurements.

Another mysterious feature of the lower bound (ii) is that it involves the $\alpha$-remainder
appearing in the norm correspondence (P1). In particular, it does not rule out existence of a quantization with the large error coefficient $\alpha$ in the norm correspondence (P1) and a small error coefficient $\beta$ in the correspondence principle (P2). At the moment, we do not know whether such
a trade-off between $\alpha$ and $\beta$ can actually happen. This leads to the following question:

\medskip
\noindent
\begin{question}\label{quest-beta} Does there exist a constant $C_4>0$ such that for every Berezin quantization
$\beta \geq C_4$?
\end{question}

\medskip
\noindent
\begin{rem}\label{rem-delta}{\rm There is no rigidity for the $\delta$-remainder in (P4). Indeed,
for the Berezin-Toeplitz quantization with the half-form correction \cite[formula (11)]{ChJSG}
the trace correspondence (P4) upgrades to
$$ \big{|} \text{trace}(T_\hbar(f)) - (2\pi\hbar)^{-n}\int_M f \frac{\omega^n}{n!}\big{|}  \leq \delta' \cdot ||f||_{L_1} \hbar^{-(n-2)}\;,$$
mind the power $-(n-2)$ in the right hand side. Thus, decreasing $\hbar$, we can make $\delta=\delta'\hbar$ arbitrarily small.
}
\end{rem}

\subsection{Bargmann space} \label{bargmann}
Even though Berezin-Toeplitz quantization of certain non-compact manifolds has been studied since
foundational paper \cite{Berezin} (see \cite{MM08} for more recent developments), no general statement
in the spirit of Theorem \ref{thm-main} is currently available in the absence of compactness.
Nevertheless, we perform a case study and explain how (P1)-(P4) extend to a noteworthy quantization of the symplectic vector space $\R^{2n}=\C^n$, namely to Toeplitz operators in Bargmann space. We also give a general estimate of the remainder in the composition of Toeplitz operators. Surprisingly, we didn't find these results in the literature. The Bargmann space serves as a source of intuition for several aspects of our exploration of compact manifolds. First, for the Bargmann space,  quasi-multiplicativity (P3) follows from an elementary (albeit tricky) algebraic consideration and thus enables one to guess the structure of the remainders in the compact case. Second, for the Bargmann space, quantization commutes with the phase-space rescaling  (see formula \eqref{eq-qcomresc} below). This highlights the significance of the rescaling which serves as a useful tool for proving rigidity of remainders in the compact case.

Recall that for any $\hbar>0$, the Bargmann space $\B_{\hbar}$ is the space of holomorphic functions of $\C^n$ which are square integrable against the weight $e ^{ - \hbar^{-1} |z|^2 } \mu $.
Here $\mu$ is the measure $|dz_1 \ldots dz_n d\con{z}_1 \ldots d \con{z}_n|$ and $|z|^2 = |z_1|^2 + \ldots + |z_n|^2$.
For any $f \in L^{\infty} ( \C^n)$, define the Toeplitz operator with multiplicator $f$
$$ T_{\hbar} (f) = \Pi_{\hbar} f : \B_{\hbar} \rightarrow \B_{\hbar}$$
where $\Pi_{\hbar}$ is the orthogonal projector of $L^{2} ( \C^n , e^{- \hbar^{-1} |z|^2} \mu )$ onto $\B_{\hbar}$. $T_{\hbar} (f)$ is a bounded operator with uniform norm $\| T_{\hbar} (f) \|_{op} \leqslant \sup |f|$. If $f \in L^{1} ( \C^n ,\mu )$, one readily checks that the operator $T_{\hbar} (f)$ is of the trace class, and (P4) holds with the vanishing error term (i.e., $\delta=0$).

For any integer $k$ and function $f : \C^n \rightarrow \C$ of class $\Cl^k$, introduce the semi-norm
\begin{gather} \label{eq:semi-norm-prime}
 | f |'_k = \sup_{|\al|=k , \; x\in \C^n} |\partial^\al f (x) |
\end{gather}
Let $\Cb^k (\C^n)$ be the space of functions $f$ of class $\Cl^k$ such that $|f|_0'$, $|f|'_1$, \ldots, $|f|'_k$ are bounded.

\begin{thm} \label{thm:toepl-barg-comp}
For any $N \in \N$, there exists $C_N >0$ such that for any $f \in \Cb ^{2N}( \C^n)$ and $g \in \Cb ^{N} ( \C^n)$, for any $\hbar$ we have
$$ T_{\hbar} (f) T_{\hbar}(g) = \sum_{\ell = 0 } ^{N-1} (-1)^{\ell} \hb^\ell \sum_{\al \in \N^n,\; |\al |= \ell} \frac{1}{\al !} T_{\hbar} \bigl( ( \partial_{z}^{\al} f)( \partial_{\con{z}}^{\al} g) \bigr) + \hbar^N R_N( f,g)\;, $$
where $\|R_{N} (f,g) \|_{op} \leqslant C_N \sum_{m=0}^{N} |f|'_{N+m} |g|'_{N - m}$.
\end{thm}

To our knowledge, the estimate of the remainder is better than what can be found in the literature. For instance, in \cite{Co} or \cite{Chi}, the number of derivatives involved in the estimates depend on the dimension $n$. However, in \cite{Le} (Section 2.4.3), Lerner obtained estimates similar to theorem \ref{thm:toepl-barg-comp} with $N=2$ for the quantity $T_1(f) T_1(g) + T_1(g) T_1(f)$.

Interestingly, it seems that Berezin-Toeplitz quantization behaves better than Weyl quantization. Indeed, in all the results we know, the number of derivatives needed to estimate the remainder when we truncate the Moyal product, depends on the dimension. Actually, the Weyl symbol of the Toeplitz operator with multiplicator $f$ is obtained by smoothing out $f$, cf. for instance \cite{BeSh}, Section 5.2 or \cite{Fo}, Section 2.7. So the higher derivatives of the Weyl symbol are controlled in a sense by the lower derivatives of the multiplicator. Another observation is that the norm estimates of pseudo-differential operators, as for instance the Calderon-Vaillancourt theorem, are generally obtained by the Cotlar-Stein lemma whereas the Schur lemma is generally sufficient to estimate the norm of Toeplitz operators. However, many authors use Cotlar-Stein lemma for Toeplitz operators even if it is not necessary.
For general results on Weyl quantization with symbols of limited regularity, we refer the reader to \cite{Sj}.

By Theorem \ref{thm:toepl-barg-comp} with $N=1$ and $N=2$
 we obtain the following version of (P3) and (P2), respectively: for any $f \in \Cb^2 ( \C^n), g \in \Cb^1 ( \C^n)$,
$$\bigl\| T_\hbar(fg) - T_\hbar(f)T_\hbar(g) \bigr\|_{op} \leq \gamma'  ( |f|_1' |g|_1' + |f|_2'|g|_0' ) \hbar $$
and for any $f \in \Cb^4 ( \C^n)$ and $g \in \Cb ^2 ( \C^n)$,
$$\Bigl\| -\frac{i}{\hbar} \bigl[ T_{\hbar}(f),T_{\hbar}(g) \bigr] - T_\hbar (\{f,g\}) \Bigr\|_{op} \leq \beta'   \bigl( |f|_2' |g|_2' + |f|_3'|g|_1' + |f|_4'|g|'_0 \bigr)\hbar\;, $$
where the constants $\gamma'$, $\beta'$ do not depend on $f$, $g$. Adapting the proof of (P2) in the closed symplectic case, we will also show that for any $f, g \in \Cb^3 ( \C^n)$,
\begin{gather} \label{eq:P2_bargmann}
 \Bigl\| -\frac{i}{\hbar} \bigl[ T_{\hbar}(f),T_{\hbar}(g) \bigr] - T_\hbar (\{f,g\}) \Bigr\|_{op} \leq \beta''   \bigl( |f|_1' |g|_3' + |f|_2'|g|_2' + |f|_3'|g|'_1 \bigr)\hbar \;.
\end{gather}
It would be interesting to calculate the constants $\alpha',\beta',\beta"$ and $\gamma'$.

Let us mention finally that by using explicit form \eqref{eq-schkerbarg} of the Schwartz kernel of $\Pi_\hbar$ and arguing as in the compact case, one gets a lower bound for $\| T_{\hbar} (f) \|_{op}$ as in (P1) provided $f$ is $C^2$-smooth and its uniform norm is attained at some point of $\C^n$.
It is unclear whether the latter condition can be relaxed.

\subsection{Joint approximate measurements}\label{subsec-joint}

In this section we present an application of our results to
joint approximate measurements of semiclassical observables. A slightly different measurement scheme
was discussed in a similar context in \cite{P-CMP} (see also \cite[Chapter 9]{P-Rosen-book}). The main novelty is that the sharp remainder bounds enable us to work on smaller scales including the quantum length scale.

\subsubsection{Preliminaries on quantum measurements}\label{subsubsec-POVMs}
We start with some preliminaries on positive operator valued measures and quantum measurements (we refer the reader to \cite{BLW},  \cite[Chapter 9]{P-Rosen-book} and references therein).
 Recall that $\cL(\Hilb)$ denotes the space of all Hermitian operators on a finite-dimensional complex
 Hilbert space $\Hilb$.

Consider a set $\Theta$ equipped with a $\sigma$-algebra $\cC$ of its subsets. An $\cL(\Hilb)$-valued {\it positive operator valued measure} (POVM) $W$ on $(\Theta,\cC)$ is
a countably additive map $W\colon \cC \to \cL(\Hilb)$ which takes each subset $X \in \cC$ to a positive operator $W(X) \in \cL(\Hilb)$ and which is normalized by $W(\Theta) = \id$.

In quantum measurement theory, $W$  represents a measuring device coupled with the system, while $\Theta$ is interpreted as the space of device readings. When the system is in a pure state $\xi \in \Hilb$, $|\xi|=1$,  the probability of finding the device in a subset $X \in \cC$ equals $\langle W(X)\xi,\xi \rangle$. Given a bounded measurable function $f\colon \Theta \to \R$, one can define the integral $\EE_W(f):=\int_\Theta f\;dW\in \cL(\Hilb)$ as follows. Introduce a measure $\mu_{W,\xi}(X) = \langle W(X) \xi,\xi\rangle$ on $\Theta$ and put  $\langle \EE_W(f)\xi,\xi\rangle = \int_\Theta f\; d\mu_{W,\xi}$, for every state $\xi \in \Hilb$. In a state $\xi$, the function $f$ becomes a random variable on $\Theta$ with respect to the measure $\mu_{W,\xi}$ with the expectation $\langle \EE_W(f)\xi,\xi\rangle$.

\medskip
\noindent
\begin{exam}\label{exm-povms-pvms}{\rm An important class of POVMs is formed by the {\it projector valued measures} $P$, for which all the operators $P(X)$, $X \in \cC$ are orthogonal projectors. For instance, every  quantum observable $A \in \cL(\Hilb)$ with $N$ pair-wise distinct eigenvalues gives rise to the projector valued measure $P:= \{P_i\}$ on the set $\Theta_N:=\{1,\ldots,N\}$ and a function $\lambda: \Theta_N \to \R$ such that $A = \sum_{i=1}^N \lambda_i P_i$ is the spectral decomposition of $A$. At a state $\xi$, the probability of the outcome
$\lambda= \lambda_i$ equals $p_i= \langle P\xi,\xi \rangle$, which agrees with the standard statistical
postulate of quantum mechanics: the observable $A$ takes value $\lambda_i$ with probability $p_i$.
Let us mention that at a pure state $\xi$ the expectation of $A$ equals $\langle A\xi,\xi \rangle$ and the variance $\mathbb{V}ar(A,\xi)$ equals $\langle A^2\xi,\xi \rangle- \langle A\xi,\xi \rangle^2$.
}
\end{exam}

\medskip
\noindent
\begin{exam}\label{exam-noise-semicl} {\rm POVMs naturally appear in the context of quantization.
This will be fundamental for our discussion on quantum measurements of semiclassical observables.
Let $\Theta=M$ be a quantizable symplectic manifold equipped with the Borel $\sigma$-algebra. Consider a Berezin quantization $( T_{\hbar} : \Cl (M , \R) \rightarrow \cL ( \Hilb_{\hbar}), \hbar \in \Lambda)$. The maps $T_{\hbar}$ being positive, by Riesz theorem, we have $\EE_{G_\hbar}(f) = T_{\hbar} (f)$ for a POVM $G_{\hbar}$ of $M$.
 Fix a sequence of quantum states $\xi_\hbar \in \Hilb_\hbar$, $|\xi_\hbar|=1$.
In this case  the measure $\mu_{G_{\hbar},\xi_\hbar}$ governs the distribution of the quantum state $\xi_\hbar$ in the phase space. Limits of such measures as $\hbar \to 0$, which are called semiclassical defect measures (or Husimi measures), has been studied in the literature, see e.g. \cite[Chapter 5]{Z}.}
\end{exam}

\medskip

A somewhat simplistic description of quantum measurement is as follows: an experimentalist,
after setting a quantum measuring device (i.e., an $\cL(\Hilb)$-valued POVM $W$ on $\Theta$), performs a measurement whose outcome, at every state $\xi$, is the measure $\mu_{W,\xi}$ on $\Theta$.
Given a function $f$ on $\Theta$ (experimentalist's choice), this procedure yields {\it an unbiased approximate measurement} of the quantum observable $A:=\EE_W(f)$. The expectation
of $A$ in every state $\xi$ coincides with the one of the measurement procedure (hence {\it unbiased}), in spite of the fact that actual probability distributions determined by the observable $A$ (see Example \ref{exm-povms-pvms} above) and the pair $(f,\mu_{W,\xi})$  could be quite different
(hence {\it approximate}). In particular, in general, the variance increases under an unbiased approximate measurement:
$
\mathbb{V}ar(f,\mu_{W,\xi}) = \mathbb{V}ar(A,\xi) + \langle \Delta_W(f)\xi,\xi\rangle\
$, where
$$
\Delta_W(f):= \EE_W(f^2)-\EE_W(f)^2
$$
is {\it the noise operator}. This operator, which is known to be positive, measures the increment of the variance. Furthermore, $\Delta_W(f)=0$ provided $W$ is a projector valued measure, and hence every quantum
observable admits a noiseless measurement in light of  Example \ref{exm-povms-pvms}.

\medskip
\noindent
\begin{exam}\label{exam-noise-semicl-1} {\rm In the setting of Example \ref{exam-noise-semicl},
the noise operator $\Delta_{G_\hbar}(f)$ is given by the expression $T_\hbar(f^2) - T_\hbar(f)^2$ appearing in the left hand side of the quasi-multiplicativity property (P3).
Look now at the case when $\xi$ is an eigenvector of $T_\hbar(f)$ with an eigenvalue
$\lambda$ for a smooth classical observable $f$ on $M$.  The expectation of $f$ with respect to $\mu_{G_{\hbar},\xi}$ equals $\lambda$, while the variance coincides with the noise $\langle \Delta_{G_\hbar}(f)\xi,\xi\rangle$. By (P3), the latter does not exceed $\gamma |f,f|_2\hbar$. It follows from the Chebyshev inequality that for every $r>0$ (perhaps depending on $\hbar$)
\begin{equation}
\label{eq-Chebyshev}
\mu_{G_{\hbar},\xi}(\{|f-\lambda| \geq r\}) \leq \frac{\gamma|f,f|_{2} \hbar}{r^2}\;.
\end{equation}
This inequality manifests the fact that in the semiclassical limit the
eigenfunctions are concentrated near the energy level $\{f=\lambda\}$, see e.g. \cite[Section 6.2.1]{Z}.
Note also that \eqref{eq-Chebyshev} provides a meaningful estimate for the
concentration at the quantum length scale $r \sim \sqrt{\hbar}$. To compare
with, the usual method to estimate the left-hand side of \eqref{eq-Chebyshev}
is to  build a local inverse of $T(f) - \lambda$ on the region $\{ |f -
\lambda| \geq r \}$. In this way, we prove that
$\mu_{G_{\hbar},\xi}(\{|f-\lambda| \geq r\})$ is in $\bigo ( \hbar^{\infty})$  for smooth $f$ and fixed $r$. More precisely, assuming that $f$ is smooth, one has for any $N$ and $r>0$
$$
\mu_{G_{\hbar},\xi}(\{|f-\lambda| \geq r\}) \leq C_N(f,r) \hbar^N \;.
$$
where $C_N(f,r)$ is a positive constant independent of $\hbar$.
}
\end{exam}

\medskip

Let $A,B \in \cL(\Hilb)$ be a pair of quantum observables. A {\it joint} unbiased approximate measurement
of $A$ and $B$ consists of an $\cL(\Hilb)$-valued POVM $W$ on some space $\Theta$  and a pair of random
variables $f$ and $g$ on $\Theta$ such that $\EE_W(f)=A$ and $\EE_W(g)=B$.

\medskip
\noindent
\begin{defin}\label{defin-min-noise}{\rm
{\it The minimal noise}
 associated to the pair $(A,B)$ is given by
$$\nu(A,B):= \inf_{W,f,g} ||\Delta_W(f)||^{1/2}_{op} \cdot ||\Delta_W(g)||^{1/2}_{op}\;,$$
where the infimum is taken over all $W,f,g$ as above. (Note: the space $\Theta$ is not fixed,
it is varying together with the POVM $W$).}
\end{defin}

\medskip
\noindent
The following {\it unsharpness principle}
(see \cite[Theorem 9.4.16]{P-Rosen-book}) provides a lower bound on
the minimal noise:
\begin{equation}
\nu(A,B) \geq \frac{1}{2}\cdot ||[A,B]||_{op}.
\end{equation}
It reflects impossibility of a noiseless joint unbiased approximate measurement of a pair
of non-commuting observables.

Fix now any scheme $T_\hbar$ of Berezin quantization of a closed quantizable symplectic
manifold $(M,\omega)$. Let $G_\hbar$ be the corresponding $\cL(\Hilb_\hbar)$-valued POVM on $M$,
i.e., $T_\hbar(f) = \int f \;dG_\hbar$ for every smooth function $f$ on $M$.
In light of Examples \ref{exam-noise-semicl} and \ref{exam-noise-semicl-1} above,
the unsharpness principle yields
\begin{equation}\label{eq-conseq-unsh}
\begin{split}
||T_\hbar(f^2) - T_\hbar(f)^2||_{op}^{1/2}  \cdot ||T_\hbar(g^2) - T_\hbar(g)^2||_{op}^{1/2}  \geq \nu(T_\hbar(f),T_\hbar(g)) \\ \geq \frac{1}{2} \cdot ||[T_\hbar(f), T_\hbar(g)]||_{op} \;\; \forall f,g \in \Cl^\infty(M)\;.
\end{split}
\end{equation}
Combining this with (P3) and (P2)
we get the following estimate for the minimal noise of a pair of semiclassical observables.

\medskip
\noindent
\begin{prop}\label{prop-double-set-meas}
\begin{equation}
\label{eq-vsp-joint}
\gamma|f,f|_2^{1/2}|g,g|_2^{1/2}\hbar \geq \nu(T_\hbar(f),T_\hbar(g)) \geq \frac{1}{2}\cdot (\hbar ||\{f,g\}||- \beta\hbar^2|f,g|_{1,3})\;.
\end{equation}
\end{prop}

\subsubsection{Joint measurements of the sign}
Fix an increasing function $u: \R \to [-1,1]$ which equals $-1$ on $(-\infty, -1]$, satisfies
$u(z)=z$ when $z$ is near $0$ and equals $1$ on $[1,+\infty)$. For a function $f$ on $M$ and a positive $s \ll 1$ the classical observable $f_s:=u(f(x)/s)$ can be considered as a smooth approximation to the sign of $f(x)$. We refer to $s$  as {\it the fuzziness} parameter.

Suppose that that the set $\{f=0,g=0\}$ is non-empty and
\begin{equation}
\label{eq-pbgz}
\sup_{\{f=0,g=0\}} |\{f,g\}|  >0\;.
\end{equation}
Assume also that $0 < s \leq t \ll 1$. In what follows we are focusing on simultaneous approximate measurements of $T_\hbar(f_s)$ and $T_\hbar(g_t)$, the semiclassical observables which  correspond to the signs of $f$ and $g$ with the fuzziness parameters $s$ and $t$, respectively. Note that the operators
$(1+ T_\hbar(f_s))/2$ and $(1+ T_\hbar(g_t))/2$ can be interpreted as ``quasi-projectors" to the phase
space regions $\{f >0\}$ and $\{g>0\}$, respectively. Since $f$ and $g$ do not Poisson commute, such a measurement is in general noisy due to the unsharpness principle.

Suppose now that $s= r\hbar^p$ and $t=R\hbar^q$ with $R,r >0$, $p \geq 0$, $q \geq 0$ and $p+q \leq 1$. The standing assumption $s \leq t$ yields $p \geq q$, and $r \leq R$ if $p=q$.

\begin{thm}\label{thm-joint-noise} Let $0 \leq q \leq p \leq 1/2$.
There exists constants $c_+ > c_- >0$ depending only on $f,g,u$ and
the metric such that
\begin{equation}
\label{eq-nu-pleq12}
c_- s^{-1}t^{-1}\hbar \leq \nu(T_\hbar(f_s),T_\hbar(g_t)) \leq c_+ s^{-1}t^{-1}\hbar\;,
\end{equation}
for all sufficiently small $\hbar$ and, in case $p=1/2$, for all sufficiently large $r$.
\end{thm}

\medskip
\noindent In particular, if $p < 1/2$, the minimal noise $\nu$ is positive and $\sim \hbar^{1-p-q}$,
and at the quantum length scale $p=q=1/2$, it is bounded from below by $c_+R^{-1}r^{-1} > 0$.

\medskip
\noindent
\begin{proof} The result immediately follows from Proposition \ref{prop-double-set-meas}.
Indeed the upper bound in \eqref{eq-vsp-joint} is
$\leq  c_1 s^{-1} t^{-1}\hbar$, and the lower bound is not less than
\begin{equation}
\label{eq-nu-minus}
c_2  s^{-1}t^{-1}\hbar - c_3  s^{-3}t^{-1}\hbar^2= c_2 \hbar s^{-1}t^{-1}(1- c_4\hbar s^{-2})\;,
\end{equation}
where $c_i$ are positive constants depending only on $f,g,u$ and the metric.
Note that the positivity of $c_2$ follows from assumption \eqref{eq-pbgz} on Poisson non-commutativity of $f$ and $g$.
Since $s^{-2}\hbar  = r^{-2}\hbar^{1-2p}$, we get that \eqref{eq-nu-pleq12} holds
for all sufficiently small $\hbar$ if $p <1/2$ and for all sufficiently small
$\hbar$ and all sufficiently large $r$ if $p=1/2$.
\end{proof}

\medskip
\noindent
This result deserves a discussion.

\medskip
\noindent
\begin{rem} \label{rem-reduced} {\rm Now we are ready to explain the advantages
of the reduced expression $|f,g|_{1,3}$ (see formula \eqref{eq-reduced}) appearing
in the remainder term of the correspondence principle (P2) as compared to $|f,g|_4$ which
includes fourth derivatives. To this end, replace for a moment $|f,g|_{1,3}$
in the remainder of (P2) by $|f,g|_4$. Then, accordingly, the lower bound \eqref{eq-nu-minus} will
be modified as
$$c_2 \hbar s^{-1}t^{-1} - c_3 \hbar^2 s^{-4}\;.$$
The first term in the right hand side is $\sim \hbar^{1-p-q}$
and the second term is $\sim \hbar^{2-4p}$. Thus for positivity of the right hand side it
is necessary that $3p-q \leq 1$. This inequality is violated, for instance, in the case when $p=1/2, q=0$, i.e., $s \sim r\hbar^{1/2}$ and $t \sim 1$. This case however can be handled by using the reduced remainder.
Indeed, inequality \eqref{eq-nu-pleq12} above yields
$\nu \sim R^{-1}r^{-1} \hbar^{1/2}$
for $p=1/2,q=0$.
}
\end{rem}

\medskip
\noindent
\begin{rem}\label{rem-pb4}
{\rm Consider the following example: Let $M=S^2=\{x^2+y^2+z^2 = 1\} \subset \R^{3}$
be the standard sphere equipped with the symplectic form of the total area $2\pi$. Put $f=x$ and $g=y$.
According to the prediction of \cite{P-CMP} (which was made for a slightly different measurement scheme)
the noise $\nu$ of such a measurement satisfies
\begin{equation}\label{eq-noise-sim}
\nu \gtrapprox \frac{\hbar}{\text{Area}(\Pi)}\;,
\end{equation}
where $\Pi$ is a ``rectangle"
$$\Pi = \{|x| \lessapprox s,\;|y| \lessapprox t,\;  z >0\} \subset S^2\;.$$
Inequality \eqref{eq-noise-sim} was proved in \cite{P-CMP} for $s,t \sim 1$.
Our methods confirm this prediction for smaller fuzziness parameters including the quantum length scales $s \sim r\hbar^{1/2}, t \sim R\hbar^{1/2}$ as well as  $s \sim r\hbar^{1/2}, t \sim 1$. Indeed, observe
that $\text{Area}(\Pi)\approx st$ and hence inequality \eqref{eq-noise-sim} follows from \eqref{eq-nu-pleq12}. Note that for $s \sim r\hbar^{1/2}, t \sim R\hbar^{1/2}$ the rectangle $\Pi$
is a ``quantum box": its area is $\sim rR\hbar$, the minimal possible (in terms of the power of $\hbar$)
area occupied by a quantum state.
}
\end{rem}

\medskip
\noindent
\begin{rem}\label{rem-yohann}{\rm
The minimal noise $\nu$  is well defined for arbitrary small
fuzziness parameters $s,t$. When $s$ and $t$ are smaller than the quantum length $\sim \sqrt{\hbar}$,
it is unclear how to calculate/estimate the minimal noise $\nu$  by
standard methods of semi-classical analysis. Indeed, the derivatives of $f_s$ and $g_t$ blow up and the remainders in (P2),(P3) dominate the leading terms. In particular, the lower bound in \eqref{eq-vsp-joint} could become negative, and hence useless. We refer  to a forthcoming paper \cite{LF-P-S} for a progress in this direction.
}
\end{rem}

\subsection{Phase space localization on quantum length scale} \label{subsec-phasespace}
Our next application of the fine structure of the remainders in the Berezin quantization
deals with phase space localization of a quantum particle at small scales. We use a model proposed
in \cite{P-CMP}.

Let $\cU=\{U_1,...,U_N\}$ be a finite open cover of a closed quantizable symplectic manifold $M$.
Given a partition of unity $\{f_1,...,f_N\}$ subordinated to $\cU$, consider the following
{\it registration procedure}: if the system is prepared in a quantum state $\xi \in \Hilb_\hbar$, $|\xi|=1$,
it is registered in the set $U_i$ with probability $\langle T_\hbar(f_i)\xi,\xi\rangle$. Here the cover
and the partition of unity may depend on $\hbar$. The registration procedure enables one
to localize a semiclassical system in the phase space.

For $x \in Q:= [-1,1]^N$ put $f_x:= \sum x_i f_i$.
Define
$$\cN_+ := \max_{x \in Q} ||T_\hbar(f_x^2) - T_\hbar(f_x)^2||_{op}$$
and
$$\cN_-:= \frac{1}{2} \cdot \max_{x,y \in Q} ||[T_\hbar(f_x),T_\hbar(f_y)]||_{op}\;.$$
Observe that $\cN_+ \geq \cN_-$ by the unsharpness principle \eqref{eq-conseq-unsh} above.

The above registration procedure is known to exhibit {\it inherent noise} which
measures the unsharpness of the registration procedure. We refer to \cite{P-CMP} and Chapter 9 of \cite{P-Rosen-book} for the precise definition. It is important for us that this noise
lies in the interval $I_{noise}:=[\cN_-,\cN_+]$, which we shall call {\it the noise interval}.
The fine remainder estimates obtained in this paper yield meaningful bounds on the noise interval
of the phase space localization procedure on small scales, up to the quantum length scale.

\medskip
\noindent
$(\clubsuit)$ {\sc The choice of the partition of unity:} To start with, let us choose a
cover of $M$ together with a subordinated partition of unity in a special way. Fix $r_0 >0$ small enough
and for $0 < r \leq r_0$ consider a maximal $r/2$-net $\{z_i\}$ of points in $M$ (with respect to the Riemannian distance $d$ associated to the metric $\rho$). This means that $d(z_i,z_j) \geq r/2$ for $i \neq j$ and $\{z_i\}$ is a maximal
collection with this property. Let $\cU$ be the cover of $M$ by metric balls $U_i:= B(z_i,r)$.
Let $u: [0,+\infty) \to [0,1]$ be a smooth cut off function which equals $1$ on $[0, 0.6]$ and vanishes on $[0.7,+\infty)$. Define  functions $g_i$ on $M$ by $g_i(x)= u(d(x,z_i)/r)$. It was shown in \cite{P-CMP} that for all sufficiently small $r < r_0(M,\omega,\rho)$
there exists $p$ (depending only on the dimension of $M$) such that every $x \in M$ is covered by at most $p$ balls $U_i$.
Moreover, the balls $B(z_i,0.6r)$ cover $M$.
Thus the functions
$$f_i:= \frac{g_i}{\sum_{i=1}^N g_i}\;,\; i=1,\dots, N$$
form a partition of unity subordinated to $\cU$ and moreover
there exists $C>0$ such that for every $r \in (0,r_0)$ and every $i=1,...,N$
\begin{equation}
\label{eq-deriv-part}
|f_i|_k \leq Cr^{-k}\;, k=1,2,3.
\end{equation}

\medskip
\noindent
In what follows we focus on the registration procedure associated to the cover $\cU$ and the
partition of unity $\{f_i\}$ described in $(\clubsuit)$ , where the radius $r \in (0,r_0]$ plays the role of a parameter. The next result provides bounds for the corresponding noise interval.

\medskip
\noindent \begin{thm}\label{thm-noise}
There exist constants $0 < c_- < c_+$ and $\kappa >0$ depending only
on $(M,\rho,\omega)$  such that
\begin{equation}\label{eq-noisint-1}
I_{noise} \subset [ c_-\hbar r^{-2},c_+\hbar r^{-2}]
\end{equation}
for any sufficiently small $\hbar >0$ and $r \in [\kappa\hbar^{1/2},r_0] $.
\end{thm}

Few remarks are in order.
Choose $R>0$, $\epsilon \in [0,1/2]$ and apply Theorem \ref{thm-noise} to $r = R \hbar ^{1/2 -\epsilon}$.
If $\epsilon = 1/2$ and $R \in (0,r_0)$ is independent of $\hbar$, then the noise
is strictly positive and of order $\sim \hbar$ as $\hbar \to 0$. This result, which was proved in \cite{P-CMP}, does not require the fine remainder estimates. The latter enter the play when $\epsilon <1/2$. Let us emphasize also that {\it for $\epsilon =0$ and a fixed $R\geq \kappa$ , i.e.,  on the quantum length scale, the noise is strictly positive and of order $\sim 1$ as $\hbar \to 0$.}

Let us mention also that the registration procedure above satisfies
{\it noise-localization uncertainty relation}:
\begin{equation}\label{eq-nl}
\text{Noise}\; \times\; \max_i \text{Size}(U_i) \geq c\hbar\;,
\end{equation}
where $\text{Size}$ is a properly defined symplectic invariant of $U_i$, and $c>0$ is independent of $\hbar$. Indeed,  $\text{Noise} \sim \hbar r^{-2}$ and since $U_i$ are Riemannian
balls of a sufficiently small radius $r$ , the size of $U_i$ is $\sim r^2$.
Relation \eqref{eq-nl} has been established in \cite{P-CMP} for the case $\epsilon = 1/2$
(i.e., for $r \sim 1$) for any partition of unity subordinated to the cover $\{U_i\}$. Here
we work on smaller scales up to the quantum length scale. As a price for that, we have to assume that the
derivatives of the functions forming the partition of unity are controlled by \eqref{eq-deriv-part}.

\medskip
\noindent
{\it Proof of Theorem \ref{thm-noise}.} Throughout the proof we denote by $c_1,c_2,...$ positive constants which
are independent on $r$ and $\hbar$. We assume that $r \leq r_0$.

Observe that by \eqref{eq-deriv-part}
\begin{equation}\label{eq-deriv-fx}
|f_x|_k \leq pC r^{-k}, \; k=1,2,3\;.
\end{equation}
Thus by (P3)
$$\cN_+ \leq c_1\hbar r^{-2} \;.$$
It has been shown in \cite[Example 4.5 and formula (28)]{P-CMP}  that
$$\mu:= \max_{x,y \in Q}||\{f_x,f_y\}|| \geq c_2 r^{-2} \;.$$
By (P1), (P2) and \eqref{eq-deriv-fx} we have that
\begin{xalignat*}{3}
\cN_- \geq  & \frac{1}{2}\cdot \mu\hbar - c_3\hbar^2 \max_{x,y \in Q} (|\{f_x,f_y\}|_2 + |f_x,f_y|_{1,3}) \\
 \geq & c_4 \hbar r^{-2}  - c_5 \hbar^{2} r^{-4}  \geq  \hbar r^{-2} ( c_4 - \frac{c_5} { \kappa^2} ) \;
\end{xalignat*}
provided $r \geq \kappa \hbar ^{1/2}$. If $\kappa$ is sufficiently large, $c_4 - c_5/ \kappa^2 = c_6 >0$.  Combining the upper bound on $\cN_+$ with the lower bound
on $\cN_-$, we get the desired result.
\qed

\medskip
\noindent
\begin{rem}
{\rm A reader with a semiclassical background has certainly recognized some symbol in exotic classes in Theorems \ref{thm-joint-noise} and \ref{thm-noise}. Recall that these symbols are smooth functions depending on $\hbar$ and satisfying an estimate of the type $|\partial^\al f | \leqslant C_\al \hbar ^{ -\delta |\al|}$ for some fixed $\delta \in [0, 1/2]$. The theory of pseudo-differential operators can be extended to these symbols, providing  an important tool in semiclassical analysis.  Here these symbols appear in the functions $f_s$, $g_t$ of Theorem \ref{thm-joint-noise} and in the functions $f_x$ of Theorem \ref{thm-noise}. Observe that (P1)-(P4) are perfectly suited to handle these exotic symbols.}
\end{rem}


\section{Constraints on the remainders}\label{sec-constraints}
In order to illustrate properties (P1)-(P4), we start with proving Theorem \ref{thm-main-1}.
Our strategy is to apply this properties to specially chosen symbols pushed to the limits of
pseudo-differential calculus, that is symbols supported in a ball of radius $\sim \sqrt{\hbar}$.
Items (i),(ii),(iii) of the theorem are proved in Sections
\ref{sec-alpha},\ref{sec-beta} and \ref{sec-gamma}, respectively.

\subsection{Test balls and scaling relations} Certain  constructions below are local, i.e., the action takes place in a neighbourhood of a point in $M$. To facilitate the discussion, we shall fix a {\it test ball} $B(r) \subset M$, that is an open ball whose closure lies in a Darboux chart equipped with coordinates $(x_1,...,x_{2n})$.
The ball $B$ is given by $\{\sum x_i^2 < r^2\}$, where $r \leq 1$.
In the chart the symplectic form $\omega$ is given by $dx_1 \wedge dx_2 +...+dx_{2n-1} \wedge dx_{2n}$.
It would be convenient to assume, without loss of generality, that the metric $\rho$ in the chart is Euclidean. This assumption will change various bounds on the norms of derivatives $|f|_N$ as well as the bounds on the quantities $\alpha,\beta,\gamma$ entering (P1)-(P3) by multiplicative constants whose precise values are irrelevant for our discussion.

In what follows every compactly supported smooth function $f \in \Cl^\infty_c(B)$
is considered as a smooth function on $M$: we extend it by $0$.

For a function $f \in \Cl^\infty_c(B(1))$
and a number $s \in (0,1]$ define a rescaled function $f_s \in \Cl^\infty(B)$ as follows:
$f_s(x)= f(x/s)$ for $x \in B(s)$ and $f_s(x)=0$ otherwise. The following obvious {\it scaling relations} turn out to be very useful below:
\begin{gather}
\label{eq-scaling_relations}
\begin{split}
|f_s|_k = s^{-k}|f|_k,\;\; |f_s,g_s|_k = s^{-k}|f,g|_k,\;\; |f_s,g_s|_{1,3} = s^{-4}|f,g|_{1,3}, \\ | \{f_s,g_s\}|_k = s^{-(k+2)} |\{f,g\}_s|_k\;.
\end{split}
\end{gather}

\subsection{$\alpha$-remainder}\label{sec-alpha}
We shall show that for a test ball $B= B^{2n}(1)$,
\begin{equation}\label{eq-alpha-bound}
\alpha \geq c \cdot \sup_f \frac{||f||^{1+1/n}}{||f||_{L_1}^{1/n} \cdot |f|_2}\;,
\end{equation}
where the supremum is taken over all smooth non-constant non-negative
compactly supported functions $f$ on $B$, and $c>0 $ is a numerical constant.
Incidentally, the finiteness of the supremum in the right hand side of \eqref{eq-alpha-bound}
follows from a generalized interpolation inequality in \cite{CZ}. Additionally, our proof shows
that the constant $c$ is independent of the dimension $2n$.

Indeed, fix any function $f$ as above.
Put $f_s(x)=f(x/s)$ with $s=\sqrt{t\hbar}$. Combining scaling relations \eqref{eq-scaling_relations}
with (P1) and (P4) we get that
$$||T_\hbar(f_s)||_{op} \geq \| f \| - \alpha t^{-1} |f|_2$$
and
$$\text{trace}(T_\hbar(f_s)) \leq (t/(2\pi))^n \cdot ||f||_{L_1}(1+\delta \hbar)\;.$$
Noticing that $||T_\hbar(f_s)||_{op} \leq \text{trace}(T_\hbar(f_s))$
since $f_s \geq 0$, we get that
$$\| f \| - \alpha t^{-1} |f|_2 \leq (t/(2\pi))^n \cdot ||f||_{L_1}(1+\delta \hbar)\;.$$
Here $t$ is fixed, and this inequality holds for all $\hbar$. Sending $\hbar \to 0$,
we conclude that
$$\alpha \geq u(t):= \frac{||f|| \cdot t - \frac{|f|_{L_1}}{(2\pi)^n} \cdot t^{n+1}}{|f|_2}\;.$$
One readily calculates that the maximal value of $u$ equals
$$c(n) \cdot  \frac{||f||^{1+1/n}}{||f||_{L_1}^{1/n} \cdot |f|_2}\;,$$
where $c(n)= 2\pi n/(n+1)^{1+1/n} \to 2\pi$ as $n \to \infty$. This proves \eqref{eq-alpha-bound}
with $c >0$ independent on $n$.
\qed

\medskip
\noindent
\begin{rem}\label{rem-alpha-dim}
{\rm By (P4), $d_\hbar:= \dim \Hilb_\hbar = (2\pi\hbar)^{-n} \cdot \text{Vol}(M) + \bigo(\hbar^{-(n-1)})$.
It turns out that a weaker dimension bound, still capturing the correct order of  $d_\hbar$
in $\hbar$, follows from the norm correspondence (P1):
\begin{equation}
\label{eq-dimension-2}
d_\hbar \geq c\alpha^{-n}\hbar^{-n}, \;\; c>0\;.
\end{equation}
The standard quantum mechanical intuition behind this formula is as follows:
consider the partition of $M$ into $\sim \text{Vol}(M)\hbar^{-n}$ ``quantum boxes", i.e., cubes of side $\sim \sqrt{\hbar}$. Since each box carries $\sim 1$ quantum state, and the states corresponding to different boxes are approximately orthogonal, $\Hilb$ contains a subspace of dimension
$\sim \hbar^{-n}$, which yields \eqref{eq-dimension-2}.  The actual proof follows this idea, with the following amendment. Instead of partitioning $M$ into quantum boxes, we cover $M$ by $\sim \hbar^{-n}$ balls of radii $\sim \sqrt{\hbar}$ as in $(\clubsuit)$ of Section \ref{subsec-phasespace} above, and apply (P1) to a specially chosen subordinated partition of unity. Let us present the formal argument.
Denote  by $c_0,c_1,...$ positive constants depending on the manifold $M$ and the metric $\rho$.
It readily follows from $(\clubsuit)$ that for every sufficiently small $r>0$, the manifold $M$
admits a partition of unity $f_1,...,f_N$ with $N \geq c_0 r^{-2n}$,
$||f_i||\geq c_1$ and $|f_i|_2 \leq c_2r^{-2}$ for all $i$.  By (P1), for all $i$
\begin{equation}\label{eq-dim-vsp}
||A_i||_{op} \geq c_1-\alpha \cdot c_2 r^{-2}\hbar\;.
\end{equation}
Since $A_i \geq 0$, we have that $\text{trace}(A_i) \geq ||A_i||_{op}$. Therefore,
$$d_\hbar= \text{trace} (\id) = \sum_{i=1}^N \text{trace} (A_i) \geq \sum_{i=1}^N ||A_i||_{op}\;.$$
Combining \eqref{eq-dim-vsp} with $N \geq c_0 r^{-2n}$ we get that
$$
d_\hbar \geq c_0 r^{-2n}(c_1-\alpha \cdot c_2 r^{-2}\hbar)\;.
$$
Choosing $r= c_3\alpha^{1/2}\hbar^{1/2}$  with $c_3 >0$ sufficiently large,
we get \eqref{eq-dimension-2}. \qed
}
\end{rem}

\subsection{$\beta$-remainder} \label{sec-beta}
Again, we work in a test ball $B= B^{2n}(1)$

{\sc Step 1:} Write $B = B^{2n}(1)$. Fix a pair of non-commuting functions $f,g \in \Cl^\infty_c(B)$. Observe that
$$||[T_\hbar(f),T_\hbar(g)]||_{op} \leq 2||T_\hbar(f)||_{op}\cdot ||T_\hbar(g)||_{op} \leq 2||f||\cdot||g||\;.$$ Let us emphasize that the inequality on the right uses positivity of
the Berezin quantization and in general fails for Weyl-like quantizations (cf. discussion after Theorem \ref{thm-main-1} above). Furthermore,
$$||T_\hbar(\{f,g\})||_{op} \geq ||\{f,g\}|| - \alpha\hbar|\{f,g\}|_2\;.$$
Combining these inequalities with the correspondence principle (P2), we get
that
\begin{equation}\label{eq-alpha-vsp-new-1}
\hbar^2(\alpha |\{f,g\}|_2 + \beta |f,g|_{1,3}) \geq \hbar||\{f,g\}||-2||f||\cdot||g||\;.
\end{equation}
The rest of the proof proceeds by two successive optimizations: first, on the ``size" of
$f,g$ through a rescaling to the quantum scale $\sim \sqrt{\hbar}$, and second, on the ``shapes"
of $f$ and $g$.

\medskip
\noindent{\sc Step 2:}
Applying the scaling relations to equation \eqref{eq-alpha-vsp-new-1} above we get that
$$\hbar^2(\alpha s^{-4} |\{f,g\}|_2 + \beta s^{-4} |f,g|_{1,3}) \geq \hbar s^{-2}||\{f,g\}||-2||f||\cdot||g||\;.$$
Put $t= s^2 \hbar^{-1}$, $a = ||\{f,g\}||$, $b= 2||f||\cdot||g||$ and rewrite this inequality
as
$$\alpha |\{f,g\}|_2 + \beta |f,g|_{1,3} \geq ta-t^2b\;.$$
The right hand side attains the maximum $a^2/(4b)$ for $t=a/(2b)$.
Note that $t=a/(2b)$ means that $s = \sqrt{a/(2b)} \cdot \hbar$, and so $s \in (0,1)$ for
$\hbar$ sufficiently small. Therefore, for all $f,g \in \Cl^\infty_c(B)$ with $\{f,g\} \neq 0$
\begin{equation}\label{eq-alpha-vsp-new-2}
\alpha |\{f,g\}|_2 + \beta |f,g|_{1,3}  \geq \frac{||\{f,g\}||^2}{8||f||\cdot||g||}\;.
\end{equation}

\medskip
\noindent{\sc Step 3:} Next, we shall take $f,g$ in the following special form.
Choose non-commuting functions $F,G \in \Cl^\infty_c(B)$, and put
$$f= z^{1/2}F\sin(z^{-1}G),\;\;g=z^{1/2}F\cos(z^{-1}G)\;,$$
where $z>0$ plays the role of a parameter. A direct calculation shows
that the Poisson bracket
$$u:= \{f,g\} = \{-F^2/2,G\}$$
is independent of $z$. At the same time, if $||F|| \leq 1$, we have that
$||f|| \cdot ||g|| \leq z$.

Recall that by Theorem \ref{thm-main-1} we have a bound $\alpha \geq C_1>0$ with $C_1$ depending
only on $(M,\omega,\rho)$. Put $Z= ||u||^2/(12C_1|u|_2)$ and observe that for all $z \in (0,Z]$
we have $|f,g|_{1,3} \leq K\cdot z^{-3}$ with some
$K >0$. Combining this with  \eqref{eq-alpha-vsp-new-2} we get that
$$\alpha |u|_2 + K\beta z^{-3} \geq z^{-1}||u||^2/8   \;\;\;\forall z \in (0,Z]\;,$$
which yields
\begin{equation}\label{eq-zo-vsp}
K\beta \geq v(z):=  z^2||u||^2/8 - z^3\alpha |u|_2 \;\;\;\forall z \in (0,Z]\;.
\end{equation}
Observe now that the function $v(z)$ attains its maximal value
$c\alpha^{-2}$ with $c= ||u||^6/(2 \cdot 12^3|u|_2^2)$
at
$$z_0= ||u||^2/(12\alpha|u|_2) \in (0,Z]\;.$$
Substituting $z_0$ into \eqref{eq-zo-vsp} we get that $\beta \geq K^{-1}c\alpha^{-2}$,
as required.
\qed

\subsection{$\gamma$-remainder} \label{sec-gamma}
Applying (P3) to $T_\hbar(fg)$ and $T_\hbar(gf)$ and subtracting we get that
$$||[T_\hbar(f),T_\hbar(g)] ||_{op} \leq 2\gamma |f,g|_2 \hbar \;.$$
On the other hand by (P1) and (P2),
$$||[T_\hbar(f),T_\hbar(g)] ||_{op} \geq \hbar ||\{f,g\}|| + O(\hbar^2)\;.$$
Combining these inequalities and letting $\hbar \to 0$, we get that
$$\gamma \geq \sup_{f,g} \frac{||\{f,g\}||}{2|f,g|_2} >0\;,$$
where the supremum is taken over all pairs of smooth non-commuting functions $f$ and $g$ on $M$.
\qed


\section{Quantization of symplectic manifolds}

\subsection{Preliminaries} \label{sec:schwartz-kernels}

Consider a compact manifold $M$ endowed with a volume
form $\mu$ and a Hermitian line bundle $A \rightarrow M$.
The space $\Cl^0 (M, A)$ of continuous sections of $A$ has a natural scalar
product $\langle \cdot, \cdot \rangle$ given by integrating the pointwise scalar product against $\mu$. We
denote by $\| \psi \| = \langle \psi, \psi \rangle^{\frac{1}{2}} $ the
corresponding norm. A bounded operator $P$ of $\Cl ^0 (M, A)$ is by definition
a continuous endomorphism of the normed vector space $(\Cl^0 (M,
A), \| \cdot \| )$. Its norm is defined by
$$ \| P \|_{op}  = \sup \frac{ \| P s \| } { \| s \| }\;, $$
where $s$ runs over the non-vanishing continuous section of $A$.
Equivalently, we could introduce the completion $L^2 ( M,A)$ of
the pre-Hilbert space $\Cl^0 (  M , A)$, $\langle \cdot, \cdot  \rangle$,
extend $P$ to a bounded operator of $L^2 ( M,A)$ and define $\|
P\|_{op}$ as the norm of the extension. Actually the space $\Cl^0 (
M , A)$ is sufficient for our needs, and we won't use its completion in the
sequel.

To any continuous section $K$ of $ A \boxtimes \con{A} \rightarrow M^2$,
\footnote{If $E \rightarrow M$ and $F\rightarrow N$ are two vector bundles, $E \boxtimes F \rightarrow M \times N$ is the vector bundle $(\pi_M ^*E) \otimes (\pi_N^* F)$, where $\pi_M$, $\pi_N$ are the projections from $M \times N$ onto $M$ and $N$ respectively.}
 we associate an endomorphism $P $ of $\Cl^0 ( M , A) $ given by
$$ (P \Psi)( x)  = \int_M K ( x, y ) \cdot \Psi ( y) \; \mu ( y)\;.$$
Here the dot stands for the contraction $A_y \times \con{A}_y \rightarrow
\C$ induced by the metric of $A$. $K$ is uniquely determined by $P$ and is
called the Schwartz kernel of $P$. $M$ being compact, $P$ is
bounded. The basic estimate we need is the Schur test, cf. for instance \cite{HaSu}, Theorem 5.2.

\begin{prop} \label{prop:norme}
Let $P$ be the endomorphism of $\Cl^0( M , A)$ with Schwartz kernel $K \in \Cl ^0 ( M^2 , A \boxtimes \con{A})$. Then $\| P \|_{op} ^2 \leqslant C_1 C_2 $ where $C_1$, $C_2$ are the non negative real numbers given by
$$ C_1 = \sup_ {x \in M} \int_M |K ( x, \cdot )| \mu\;, \quad C_2 = \sup_ {y \in M} \int_M |K (  \cdot, y  )| \mu\;.$$
\end{prop}


We will also need the following easy properties. Let $K \in  \Cl^0 ( M^2 , A
\boxtimes \con{A})$ be the Schwartz kernel of $P$.
\begin{itemize}
\item For any $  f \in \Cl ^0 (M,
\C)$, $(1 \boxtimes f) K$ and $(f \boxtimes 1) K$ are respectively the
Schwartz kernels of $P f$   and $f P$ respectively.
\item
Let $\nabla$ be a Hermitian connection of $A$ and assume that $K$ is of
class $\Cl ^1$. Then for any continuous vector field $X$ of $M$, $( \nabla_X
\boxtimes \op{id}) K$ is the Schwartz kernel of $\nabla_X \circ P$. Furthermore, if $X$ is of class
$\Cl^1$, the operator $ P \circ \nabla_X : \Cl ^1 (M, A) \rightarrow \Cl^0 (M, A)$ extends to the bounded operator of $\Cl^0 (M, A)$ with kernel $- ( \op{id} \boxtimes ( \nabla_X + \op{div} X)) K$. Here the divergence is defined by the equality: $\mathcal{L}_X \mu = \op{div} (X) \mu$.
\end{itemize}

{\em In the sequel, we often denote an operator and its Schwartz kernel by
  the same letter.}

\subsection{Toeplitz operators} \label{sec:toeplitz-operators}

Consider as in Section \ref{sec:schwartz-kernels} a compact manifold $M$
endowed with a volume form and a Hermitian line bundle $A \rightarrow M$.
Let  $\Hilb$ be a finite dimensional subspace of $\Ci ( M , A)$. Let $B$
be the section of $A \boxtimes \con{A}$ defined by
\begin{gather} \label{eq:defB}
B ( x, y) = \sum_{i=1}^{N} e_{i} ( x) \otimes \overline{e}_i (y), \qquad x,y \in M
\end{gather}
where $( e_{i}, \; i=1, \ldots , N)$ is any orthonormal basis of
$\Hilb$. The operator $\Pi$ with Schwartz kernel $B$ is the projector from $\Cl ^0 (
M, A)$ onto $\Hilb$ with kernel the orthogonal complement of $\Hilb$ in
$\Cl ^0 ( M , A)$. Even if we are not in a genuine Hilbert space, we call
$\Pi$ an orthogonal projector. For any $f \in \Cl^0 (M)$, define the Toeplitz operator
$$ T( f) := \Pi f : \Hilb \rightarrow \Hilb\;. $$
Here $f$ stands for the multiplication operator by $f$. The map sending $f$
to $T(f)$ is clearly linear and positive. Furthermore $T(1) = \op{id}$.

\subsection{Bergman kernels and generalisations}

Consider a quantizable symplectic compact manifold $(M, \om)$. Our aim is to produce a Berezin quantization $ (T_{\hbar} : \Cl^0 (M) \rightarrow \cL ( \Hilb_{\hbar}) , \hbar \in \Lambda)$. We will use the integer parameter $k $ instead of $\hbar \in \N$ \footnote{$\N$ is the set $\Z_{\geq 0}$ of non negative integers.}
having in mind that $\hbar = 1/k$. The Hilbert space $\Hilb_k$ will be defined as a finite dimensional subspace of $\Ci (M, A_k)$ with $A_k$ a conveniently defined Hermitian line bundle. The linear map $T_k :  \Cl^0 (M) \rightarrow \cL ( \Hilb_{k}) $ will be the corresponding Toeplitz quantization as in Section \ref{sec:toeplitz-operators}.

$M$ being quantizable, it admits a  prequantum
bundle $L$, that is a Hermitian line bundle
endowed with a connection $\nabla$ of curvature $\frac{1}{i} \om$. Consider a complex structure $j$, which is not necessarily integrable, but compatible with $\om$ meaning that for any tangent vectors $X,Y \in T_{p}
M$, $ \om ( jX, jY) = \om (X, Y)$ and if $X$ does not vanish, $\om ( X, jX)
>0$. We denote by $T^{1,0}M$ the subbundle $\ker ( j - i ) $ of $TM \otimes \C$.
Consider also an auxiliary Hermitian line bundle $A$.

For any $k \in \N$, let $A_ k = L^k \otimes A$ and endow  the space
$\Cl^0 (M, A_k)$  with the scalar product defined by integrating the
pointwise scalar product against the Liouville volume $\mu = \om^n / n
!$. To any finite dimensional subspace $\Hilb_k$ of $\Ci (M, A_k)$ is associated a
smooth kernel $B_k \in \Ci ( M^2, A_k \boxtimes \con{A}_k)$ defined as in
(\ref{eq:defB}). The metrics of $L$ and $A$ induce identifications $L_x \otimes \con{L}_x \simeq \C$ and $A_x \otimes \con{A}_x \simeq \C$. In the sequel, we will often view $B_k (x,x)$ as a complex number through these identifications.

\begin{thm} \label{theo:B_kernel}
There exists a family $( \Hilb_k \subset \Ci ( M , L^k \otimes A) , \; k \in \N)$ of finite dimensional subspaces such that the corresponding family $(B_k)$ satisfies for any $m \in \N$,
\begin{gather} \label{eq:kerprojdevas}
 B_k ( x,  y ) = \Bigl( \frac{k}{2 \pi} \Bigr) ^n E^k ( x, y) \sum_{\ell \in \Z \cap [ -m , m/2]} k^{-\ell} \si_{\ell} ( x, y) + \bigo_{\infty} ( k^{n-(m+1)/2})\;,
\end{gather}
where $2n$ is the dimension of $M$ and
\begin{itemize}
\item $E$ is a section of $L \boxtimes \con{L}$ satisfying $E(x,x) = 1$, $|E (x, y )| < 1$ if $ x \neq y$ and for any vector field $Z \in \Ci (M ,T^{1,0} M)$, $(\nabla_{\con{Z}} \boxtimes \op{id} )E$ and $( \op{id} \boxtimes \nabla_Z ) E$ vanish to second order along the diagonal of $M^2 =M \times M$.
\item For any $\ell \in \Z$, $\si_\ell$ is a section of $A \boxtimes \con{A}$. If $\ell$ is negative, $\si_\ell$ vanishes to order $-3\ell$ along the diagonal.
\end{itemize}
Furthermore $\si_{0} (x, x)  = 1$ for any $x \in M$.
\end{thm}

The notation $\bigo_{\infty} ( k^{N} )$ has been introduced by the first author in previous papers and refers to a uniform control of the section and its successive derivatives. The precise meaning is as follows. A family $(\Psi( \cdot, k) \in \Ci (M^2, A_k \boxtimes \con{A}_k), \; k \in \N)$ is in $\bigo_{\infty} ( k^{N} )$ if for any open set $U$ of $M^2$, for any compact subset $K$ of $U$, for any unitary frames $\tau_A: U \rightarrow A \boxtimes \con{A}$ and  $\tau_{L} : U \rightarrow L \boxtimes \con{L}$, for any $m\in \N$, for any vector fields $X_1$, \ldots, $X_m$ of $M^2$, there exists $C>0$ such that for any $k$,
\begin{gather} \label{eq:defbigoinf}
 \Psi( \cdot, k) = f_k \tau_{L}^k \otimes \tau_A \text{ on $U$ } \Rightarrow \op{Sup}_K | X_1 \ldots X_m f_k | \leqslant C k^{ N + m} .
\end{gather}
Observe that one loses a factor $k$ at each derivative, so that condition (\ref{eq:defbigoinf}) does not depend on the choice of the frame $\tau_{L}$.

It is not difficult to see that for any $\si \in \Ci ( M^2, A \boxtimes \con{A})$ vanishing to order $p$ along the diagonal, the familly $( E^k \otimes \si, k \in \N)$ is in $\bigoinf ( k^{ -p/2})$. So in Theorem \ref{theo:B_kernel}, the family $(k^{-\ell} E^k \otimes \si_{\ell})$ is in $\bigoinf ( k^{-\ell})$ if $\ell$ is non negative and in $\bigoinf ( k^{\ell/2})$ if $\ell$ is negative. We refer the reader to Sections 2.2 and 2.3 of \cite{oim_symp} for more details and other basic properties of the $\bigo_{\infty} ( k^{N} )$.

In the K\"ahler case, that is when $j$ is integrable and $L$, $A$ are holomorphic line bundles, we can define $\Hilb_k$ as the space of holomorphic sections of $A_k$. The corresponding kernel $B_k$ is called the Bergman kernel. The asymptotic of $B_k$ given in Theorem \ref{theo:B_kernel} has been deduced in \cite{oim_op} (Corollary 1) from the seminal paper \cite{BoSj}. A direct proof has been given in \cite{BeBeSj}, cf. also \cite{MaMa} and \cite{ShZe} for similar results. In this case, we can even choose $E$  in such a way that the $\si_{\ell}$'s with negative $\ell$ are identically null.

In the general symplectic case, the spaces $\Hilb_k$ are defined in
such a way that $B_k$ admits an asymptotic expansion of this form. The
existence of such a quantization has been proved in \cite{oim_symp} using
the ideas of \cite{BoGu}, cf. also \cite{MaMa} and \cite{ShZe} for similar
results. In the construction proposed in \cite{oim_symp}, we start with any
sections $E$ and $\si_0$ satisfying the asumptions of Theorem
\ref{theo:B_kernel}. We assume also that $\con{E} (x,y) = E(y,x)$ and $\con{\si}_0 (x,y) = \si_0 ( y,x)$, so that the operator $P_k$ with Schwartz kernel $\bigl( \frac{k}{2 \pi} \bigr)^n E^k \si_0 $ is self-adjoint. One proves that the spectrum $P_k$ concentrates onto 0 and 1, in the sense that $$ \op{spec} (P_k) \subset [-Ck^{-1/2} , C k^{-1/2}] \cup [ 1 - C k^{-1/2}, 1+ C k^{-1/2}]\;,$$
where $C$ is a positive constant independent of $k$. Furthermore, for any
$k$, $\op{spec} ( P_k)  \cap [ 1 - C k^{-1/2}, 1+ C k^{-1/2}]$ consists of
a finite number of eigenvalues, each having a finite multiplicity and the
corresponding eigenvectors are smooth. We define $\Hilb_k$ as the sum of the corresponding eigenspaces
$$ \Hilb_k := \bigoplus_{\la \in \op{spec} ( P_k)  \cap [ 1 - C k^{-1/2}, 1+ C k^{-1/2}]} \ker ( P_k - \la)\;.$$
Then one proves that the corresponding kernel has the expected behaviour.

\subsection{Berezin-Toeplitz operators, (P1) and (P4)}

Consider a family  $( \Hilb_k \subset \Ci ( M , L^k \otimes A) , \; k \in \N)$ satisfying the conditions of theorem \ref{theo:B_kernel}.
For any $f \in \Cl^0 (M)$, define the Toeplitz operator
$$ T_k( f) := \Pi_k f : \Hilb_k \rightarrow \Hilb_k\;. $$
where $\Pi_k$ is the orthogonal projector of  $ \Cl^0 ( M ,A_k)$  onto $\Hilb_k$, as in Section \ref{sec:toeplitz-operators}.
We shall show that this construction satisfies properties (P1)-(P4) of Theorem \ref{thm-main}. Let us start with the norm and trace correspondences, since their proofs are very short.

For the estimation of the norm,  we will use special vectors of $\Hilb_k$ called coherent states. Let $B_k$ be the Schwartz kernel of $\Pi_k$.  Let $x \in M$ and $u, v$ be unitary vectors of $L_x$ and $A_x$ respectively. Let $\Psi _k$ be the section of $L^k \otimes A$ defined by
\begin{gather} \label{eq:def_psi_k}
 \Psi_k ( y) = B_k ( y,x) \cdot (u^k\otimes v), \qquad \forall y \in M\;,
\end{gather}
where the dot stands for the contractions $A_{k,y} \otimes \con{A}_{k,x} \otimes A_{k,x} \rightarrow A_{k,y}$ induced by the metrics of $L$ and $A$. Expanding $B_k$ in an orthonormal basis $(e_{k,i}, i =1, \ldots , N_k)$ as in (\ref{eq:defB}), we see that $\Psi_k$ belongs to $\Hilb_k$. Furthermore
\begin{gather} \label{eq:norm_psi_k}
 \| \Psi_k  \|^2 =  \sum_{i=1}^{N_k} | e_{k,i} (x) |^2 = B_k (x,x)\;,
\end{gather}
where we view $B_k(x,x)$ as a number as explained before Theorem \ref{theo:B_kernel}. We deduce from Theorem \ref{theo:B_kernel} that $\| \Psi_k \|^2 \sim (k / 2 \pi)^n$. When $k$ is sufficiently large, we set $\Psi_k^{\op{n}} = \Psi_k / \| \Psi_k \|$.

\begin{prop}
There exists $\al>0$ such that for any $f \in \Cl^2 (M)$ whose $x$ is a critical point, we have for any $k$
\begin{gather*}
 \bigl\| T_k (f) \Psi_k^{\op{n}}  - f(x) \Psi_k^{\op n}  \bigr\| \leqslant \al k^{-1} |f|_2\;.
\end{gather*}
Furthermore $\al$ does not depend on  $x$, $u$ and $v$.
\end{prop}
Applying this to a point $x$ where $|f|$ attains its maximum, we deduce that the spectrum of $T_k (f)$ intersects $\|f \| + \al k^{-1} |f|_2 [ -1, 1]$.  This implies the property $(P1)$.

\begin{proof}
Let $\la = f(x)$. Let $(U,y_i)$ be a coordinate system centered at $x$. Let $V$ be a relatively compact open neighborhood of $x$ contained in $U$. Write $\delta = \sum y_i^2$. Then if $x$ is a critical point of $f$, we have
\begin{gather}  \label{eq:1}
| f (y) - \la| \leqslant C_1 |f|_2 \delta (y) \qquad \forall y \in V\;,
\end{gather}
 where $C_1$ does not depend on $y$ and $f$. By Theorem \ref{theo:B_kernel}, we have
\begin{gather} \label{eq:estim_classique}
 \int_V | \Psi^{\op n } _k  |^2 \delta ^2 \mu =  \bigo(k^{-2})\;, \qquad \int_{M\setminus V} | \Psi^{\op n } _k  |^2  \mu = \bigo( k^{-\infty})\;.
\end{gather}
Indeed, we can adapt the standard proof of the K\"ahler case as follows. Recall that $\| \Psi_k \| ^2 \sim (k/2\pi)^n$.  Furthermore there exists $0< r< 1$, such that for any $y \in M \setminus V$, we have $|E (y,x) | \leqslant r$. The second estimate of (\ref{eq:estim_classique}) follows easily from Theorem \ref{theo:B_kernel}. For the first one, we use that there exists $C_2>0$ such that for any $y \in V$, $|E(y,x) | \leqslant  e^{ -\delta (y)/C_2}$. So by Theorem \ref{theo:B_kernel}, there exists $C_3>0$ such that
$$| \Psi_{k}^{\op n } (y) | \leqslant k^{n/2} C_3 e^{-k \delta (y)/C_2}, \qquad \forall y \in V $$ on $U$. Write $\mu = g dy_1 \wedge \ldots \wedge dy_{2n}$ and let $C_4>0$ be such that $|g | \leqslant C_4$ on $V$. We have
\begin{xalignat*}{2}
 \int_{U} | \Psi_k^{\op{n}} |^2  \delta^2 \mu & \leqslant  k^{n} C_3^2 C_4 \int_{\R^{2n} } e^{ - 2 k |u|^2/C_2 } |u|^4 \; du \\  & = k^{-2} C_3^2  C_4 \int_{\R^{2n} } e^{ - 2 |u|^2/C_2 } |u|^4 \; du\;.
\end{xalignat*}
This proves the first equation of  (\ref{eq:estim_classique}).

Now, using Equations (\ref{eq:estim_classique}), (\ref{eq:1}) and the fact that $| f(y)  - \la | \leqslant 2 |f|_2$ on $M$, we obtain that
$$ \| (f - \la) \Psi_k^{\op{n}} \| ^2 = \int _M |f(y) - \la|^2 | \Psi_k^{\op{n}}(y)  |^2 \mu (y) \leqslant C k^{-2} |f|_2^2 $$
for some $C>0$ independent of $f$. Since $\| \Pi_k \|_{op} \leqslant 1$, it follows  that
$$ \| \Pi_k f  \Psi^{\op n}_k - \lambda \Psi^{\op n} _k \| \leqslant \al k^{-1} |f|_2\;,$$
where $\al = C^{1/2}$. The fact that $\al$ may be chosen independently of $x$, $u$ and $v$, follows from the compactness of $M$.
\end{proof}

Let us prove property (P4).

\begin{prop}
For any $k$, there exists a sequence $(\rho ( \cdot, k))$ in $\Ci (M)$ such that for any $f \in \Cl^0 (M)$
$$\op{tr} (T_k (f))  = \Bigl( \frac{k}{2 \pi } \Bigr) ^n \int_M f \rho ( \cdot, k) \mu\;, $$
where $\mu$ is the Liouville volume. Furthermore, $\rho(\cdot, k) = 1 + \bigo( k^{-1})$ uniformly on $M$.
\end{prop}

\begin{proof}
Denote by $h_k$ the metric of $A_k$. Then
\begin{xalignat*}{2}
 \op{tr} (T_k (f))  = & \sum_i \langle f e_{k,i} , e_{k,i} \rangle
 =  \sum_i \int_M f(x) h_k ( e_{k,i} ( x) , e_{k,i} (x) ) \mu ( x) \\
=  &\int_M f(x) B_k (x,x) \mu (x)\;,
\end{xalignat*}
where we identify $B_k(x,x)$ with a number as previously.
By Theorem \ref{theo:B_kernel}, we know that $B_k (x,x) = \bigl( k / 2 \pi \bigr)^n \rho (x,k)$, where $\rho ( \cdot , k)$  has the asymptotic expansion $1+ k^{-1} \si_1 (x,x) + k^{-2} \si_2(x,x) + \ldots$.
\end{proof}


\subsection{Proof of sharp remainder estimates, (P2) and (P3)}

Our strategy is to make a detour through the  Kostant-Souriau operators and the corresponding Toeplitz operators, which are well-behaved in terms of commutator estimates. In particular these modified Toeplitz operators satisfy a correspondence principle with a remainder better than (P2), involving only second derivatives.
We will then analyse the Toeplitz operators as perturbation of the former.

\subsubsection{Kostant-Souriau operators}

Let us introduce a covariant derivative $\nabla^A$ of $A$. We denote by $\nabla^k$ the covariant derivative of $L^k \otimes A$ induced by $\nabla^A$ and the covariant derivative $\nabla$ of $L$. Let $f \in \Cl ^1 (M)$ and denote by $X$ its Hamiltonian vector field.
\footnote{In this paper the Hamiltonian vector field $X_f$
of a function $f$ is defined by $i_{X_f}\omega  + df = 0$, and the Poisson bracket is given by $\{f,g\}= -\omega(X_f,X_g)$.} The Kostant-Souriau operator associated to $f$ acting on sections of $L^k \otimes A$ is given by
\begin{gather} \label{eq:def_Kostant_souriau}
H_k (f) =  f + \tfrac{1}{i k} \nabla^k_X\;.
\end{gather}
It was discovered independently by Kostant \cite{Ko} and Souriau \cite{So} that when $A$ is the trivial line bundle and $\nabla^A $ the de Rham derivative, $H_k$ satisfies an exact correspondence principle. For a general pair $(A, \nabla^A)$, we have for any $f, g \in \Cl^2 (M)$
\begin{gather}   \label{eq:Kostant_Souriau}
[H_k ( f) , H_k ( g) ] = \tfrac{i}{k} H_k ( \{ f, g \} ) - \tfrac{1}{k^2} \Om_A ( X, Y)\;,
\end{gather}
where $\Om_A$ is the curvature of $\nabla^A$.

When $f$ is of class $\Cl^2$, $H_k (f)$ sends $\Cl^1 (M, A_k)$ into $\Cl^0
(M, A_k)$ so the same holds for the commutator  $[ H_k (f) , \Pi_k ]$. By the properties recalled after Proposition \ref{prop:norme}, this commutator extends to a bounded
operator of $(\Cl ^0 (M, A_k), \| \cdot \|)$.
  When $f$ is smooth, it was proved in \cite{oim_symp} that the norm of $[ H_k (f) , \Pi_k ]$ is a $\bigo (k^{-1})$. We will extend this to functions of class $\Cl^2$ and prove that the $\bigo ( k^{-1})$ only depends on the $\Cl^2$ norm of $f$.

\begin{thm} \label{theo:compi} There exists $C>0$ such that for any $f \in \Cl^2 (M)$, we have for any $k \in \N$,
 $$ \bigl\| [ H_k (f) , \Pi_k ] \bigr\|_{op}  \leqslant C  k^{-1}  | f | _2\;. $$
\end{thm}

The proof will be given in Section \ref{sec:proofs}. It is a consequence of Theorem \ref{theo:B_kernel}.  Denote by $T_k^c (f)$ the operator
\begin{gather} \label{eq:Toeplitz_correction}
T_k ^c (f) = \Pi_k H_k (f) :  \Hilb_k \rightarrow \Hilb_k\;.
\end{gather}
The superscript $c$ stands for correction. Surprisingly, we only need to assume $f$ and $g$ of class $\Cl^2$ to get the sharp correspondence principle for $T_k^c$.

\begin{prop} \label{prop:P2_correction}
We have for any $f$ and $g$ in $\Cl ^2 ( M)$,
\begin{gather} \label{eq:P2_cor}
 [ T_k^c ( f) , T_k^c ( g) ] = \tfrac{i}{k} T_k^c ( \{ f, g \} ) + \bigo ( k^{-2}) |f|_2 |g|_2\;.
\end{gather}
\end{prop}

Here it is implicitly meant that the $\bigo ( k^{-2})$'s do not depend on $f$ or $g$. More precisely, the $\bigo ( k^{-2} )$ is a term whose uniform norm is $\leq Ck^{-2}$, where $C$ depends only on the family $( \Hilb_k)$, but {\bf not} on $f$ or $g$. We use the same convention in the sequel.

\begin{proof}
We  check by a straightforward computation that
$$ \Pi_k [ \Pi_k , H_k(f) ] [ \Pi_k , H_k(g) ] \Pi_k = T^c_k(f) T_k^c (g) - \Pi_k H_k(f) H_k(g) \Pi_k\;.$$
By Theorem \ref{theo:compi}, the left hand side is a $ \bigo ( k^{-2} )|f|_2 |g|_2$. So we have that
$$ [T^c_k(f), T_k^c (g)] = \Pi_k [H_k ( f) , H_k ( g) ] \Pi_k + \bigo ( k^{-2} )|f|_2 |g|_2\;. $$
Using Kostant-Souriau formula (\ref{eq:Kostant_Souriau}) and the fact that $\Pi_k \Om_A ( X, Y) \Pi_k = \bigo (1) |f|_1 |g|_1$, we get (\ref{eq:P2_cor}).
\end{proof}

\subsubsection{K\"ahler case}

We assume in this section that $(M, \om , j)$ is a K\"ahler manifold, $L$, $A$ are holomorphic Hermitian line bundles over $M$, and the connections $\nabla$ and $\nabla^A$ are the Chern connections. Furthermore, $\Hilb_k$ is the space of holomorphic sections of $A_k $.

\begin{lemma} \label{lem:cor}
For any vector field $X$ of $M$ of class $\Cl^1$, we have
$$\Pi_k \nabla^k_X \Pi_k = - \Pi_k \op{div} (Z) \Pi_k\;, $$
where $Z = \tfrac{1}{2} (X - i j X)$ and $\op{div} (Z)$ is the divergence of $Z$ with respect to the Liouville form.
\end{lemma}

\begin{proof}
 Since $\Hilb_k$ consists of holomorphic sections and $\con{Z}$ is a section of $T^{0,1}M$, $\Pi_k \nabla^k_{\con{Z}} \Pi_k = 0$. Since $Z$ is of class $\Cl^1$, the integral $\int \mathcal{L}_Z ( f \mu) $ vanishes for any smooth function $f$. We obtain that for any $s$, $ t \in \Ci ( M , A_k)$,
$$ \langle \nabla^k_{Z} s, t \rangle + \langle s, \nabla^k_{\con{Z}}  t \rangle + \langle \op{div} (Z)  s, t \rangle = 0 . $$
Applying this to $s$, $t \in \Hilb_k$, we deduce that   $\Pi_k ( \nabla_Z + \op{div}( Z) ) \Pi_k =0$. Consequently
$$ \Pi_k \nabla^k_X \Pi_k = \Pi_k \nabla^k_{Z} \Pi_k + \Pi_k \nabla^k_{\con{Z}} \Pi_k = - \Pi_k \op{div} (Z) \Pi_k\;,$$
which was to be proved.
\end{proof}

When $X$ is the Hamiltonian vector field of $f \in \Cl ^2 (M)$, we have $\op{div} X = 0$ so that $\op{div}(Z) = - i/2 \op{div} (jX) = i \Lap f$ where $\Lap$ is the holomorphic Laplacian.   We deduce  Tuynman's formula \cite{Tu}:
\begin{gather} \label{eq:tuynman}
  \Pi_k i  \nabla_X \Pi_k = \Pi_k (\Lap f) \Pi_k\;.
\end{gather}
Recall that $T_k (f) $ is the Toeplitz operator $\Pi_k f : \Hilb_k \rightarrow \Hilb_k$.  By (\ref{eq:tuynman}), we have
\begin{gather} \label{eq:corr}
T^c_k (f) = T_k (f) - \tfrac{1}{k} T_k ( \Lap f) = T_k ( f) + \bigo( k^{-1}) |f|_2\;.
\end{gather}
Let us prove that $T_k$ satisfies the quasi-multiplicativity (P3).
\begin{prop} \label{prop:P1_kahler}
For any functions $f \in \Cl^1 (M)$ and $g \in \Cl ^2 ( M)$, we have
\begin{gather*}
 T_k ( f) T_k ( g) =T_k ( fg) + \bigo ( k^{-1})(|f|_0|g|_2 + |f|_1 |g|_1)\;, \\
 T_k ( g) T_k (f) = T_k ( fg ) + \bigo ( k^{-1})(|f|_0|g|_2 + |f|_1 |g|_1)\;.
\end{gather*}
\end{prop}

\begin{proof} Let $Y$ be the Hamiltonian vector field of $g$. We have
 $$  \Pi_k f [ \Pi_k , H_k (g) ] \Pi_k =  T_k ( f) T_k ( g) - T_k (fg) + \Pi_k f \Pi_k \tfrac{1}{ik} \nabla_Y^k \Pi_k - \Pi_k  \tfrac{1}{ik} \nabla_{fY}^k \Pi_k\;. $$
By Theorem \ref{theo:compi}, the left hand side is a $\bigo ( k^{-1} ) |f|_0 |g|_2$.
By Lemma \ref{lem:cor}, $\Pi_k  \tfrac{1}{ik} \nabla_Y^k \Pi_k = \bigo ( k^{-1}) |Y|_1$, so that
\begin{gather} \label{eq:16}
\begin{gathered}
\Pi_k f \Pi_k \tfrac{1}{ik} \nabla_Y^k \Pi_k = \bigo( k^{-1} ) |f|_0 |g|_2\;, \\
 \Pi_k  \tfrac{1}{ik} \nabla_{fY}^k \Pi_k = \bigo ( k^{-1} ) |fY|_1 = \bigo(k^{-1}) (|f|_0 |g|_2 + |f|_1 |g|_1)\;.
\end{gathered}
\end{gather}
which concludes the proof of the first equation. To get the second one, we take the adjoint.
\end{proof}

Finally let us show the sharp correspondence principle (P2).

\begin{prop} \label{prop:P2}
We have for any $f,g \in \Cl^3 (M)$,
\begin{gather*}  [ T_k ( f) , T_k ( g) ] = \tfrac{i}{k} T_k ( \{ f, g \} ) + \bigo ( k^{-2}) ( |f|_1 |g|_3 + |f|_2 |g|_2 + |f|_3|g|_1 )\;.
\end{gather*}
\end{prop}

\begin{proof}
By Proposition \ref{prop:P1_kahler}, we have that for any $u \in \Cl^1 (M)$ and $v \in \Cl^ 2(M)$, $[T_k (u), T_k (v) ] = \bigo ( k^{-1} ) ( |u|_0 |v|_2 + |u|_1 |v|_1)$. Consequently,
$$ [ T_k ( \Lap f) , T_k (g) ] = \bigo( k^{-1}) ( |f|_2 |g|_2 + |f|_3 |g|_1 )\;.$$
Similarly,
$$ [T_k (f) , T_k ( \Lap g) ] = \bigo ( k^{-1} ) (|f|_1 |g|_3 + |f|_2 |g|_2)\;.$$
Using that $T_k (u) = \bigo (1) |u|_0$, we have
$$ [T_k ( \Lap f) , T_k ( \Lap g) ] = \bigo (1) |f|_2 |g|_2\;, \qquad T_k (\Lap \{ f, g \} ) = \bigo (1) |f,g|_3\;.$$
We conclude with  Proposition \ref{prop:P2_correction}  by using $T^c_k (f) = T_k (f) - \tfrac{1}{k} T_k ( \Lap f)$.
\end{proof}

\subsubsection{Symplectic case}

Let us return to the general symplectic case.
We do not know how to generalize Lemma  \ref{lem:cor}. Instead we will use the following result.
\begin{thm} \label{lem:Tu_symp}
There exists $C>0$ such that for any $Z \in \Cl^0 ( M , T^{1,0} M)$,
\begin{gather} \label{eq:12}
 \forall k \in \N, \qquad \bigl\| \tfrac{i}{k} \nabla^k_{\con{Z}} \Pi_k  \bigr\|_{op} \leqslant C k^{-1} | Z|_0\;,
\end{gather}
and if $Z$ is of class $\Cl^1$, for any $f \in \Cl ^2 (M)$, $[H_k (f) ,
\tfrac{i}{k} \nabla^k_{\con{Z}} \Pi_k ] : \Ci ( M , A_k) \rightarrow \Cl^0
( M , A_k)$ extends continuously  to a bounded operator of $(\Cl^0 ( M ,
A_k), \| \cdot \|_k )$ satisfying
\begin{gather} \label{eq:13}
\forall k \in \N, \qquad  \bigl\| [ H_k ( f) , \tfrac{i}{k} \nabla^k_{\con{Z}}  \Pi_k ] \bigr\|_{op} \leqslant C k^{-2} (| f|_2 | Z|_0 + |f|_1 |Z|_1)\;.
\end{gather}
\end{thm}
The proof will be given in Section \ref{sec:proofs}.
A consequence of the first inequality is the following lemma.
\begin{lemma} \label{lem:estim_norm_cor}
We have for any $X \in \Cl ^1 (M, TM)$ that
$$ \Pi_k \tfrac{i}{k} \nabla^k_X \Pi_k = \bigo ( k^{-1} ) |X|_1\;.$$
\end{lemma}

\begin{proof}
Write $X = Z + \con{Z}$ with $Z$ a section of $T^{1,0}M$.
By (\ref{eq:12}), we have
\begin{gather} \label{eq:14}
\Pi_k \tfrac{i}{k} \nabla^k_{\con{Z}} \Pi_k = \bigo ( k^{-1} ) |Z|_0\;.
\end{gather}
Taking the adjoint, we get
$\Pi_k \tfrac{i}{k} ( \nabla^k_{Z} + \op{div} (Z) ) \Pi_k = \bigo ( k^{-1} ) |Z|_0$. Since $\op{div} (Z) $ is a $\bigo ( |Z|_1)$ in uniform norm, we obtain
\begin{gather} \label{eq:15}
\Pi_k\tfrac{i}{k}  \nabla^k_{Z} \Pi_k = \bigo (k^{-1} ) |Z|_1\;.
\end{gather}
Adding (\ref{eq:14}) and (\ref{eq:15}), we get the result.
\end{proof}

As a consequence, we have
\begin{gather} \label{eq:4}
T_k ( f) = T_k^c (f) + \bigo ( k^{-1}) |f|_2\;.
\end{gather}
We also deduce the quasi-multiplicativity (P3).
\begin{prop} \label{prop:P1}
For any functions $f \in \Cl^1 (M)$ and $g \in \Cl ^2 ( M)$, we have
\begin{gather*}
 T_k ( f) T_k ( g) =T_k ( fg) + \bigo ( k^{-1})(|f|_0|g|_2 + |f|_1 |g|_1)\;, \\
 T_k ( g) T_k (f) = T_k ( fg ) + \bigo ( k^{-1})(|f|_0|g|_2 + |f|_1 |g|_1)\;.
\end{gather*}
\end{prop}

\begin{proof} The proof is exactly the same as the one of Proposition \ref{prop:P1_kahler} except that we deduce Equations (\ref{eq:16}) from Lemma \ref{lem:estim_norm_cor} instead of Lemma \ref{lem:cor}.
\end{proof}

Finally we show the sharp correspondence principle (P2).

\begin{prop} \label{prop:commut_correction}
For any  $f,g \in \Cl^3 (M)$,
\begin{gather*}  [ T_k ( f) , T_k ( g) ] = \tfrac{i}{k} T_k ( \{ f, g \} ) + \bigo ( k^{-2}) ( |f|_1 |g|_3 + |f|_2 |g|_2 + |f|_3|g|_1 )\;.
\end{gather*}
\end{prop}

\begin{proof} Denote by $X$ and $Y$ the Hamiltonian vector fields of $f$ and $g$.
By (\ref{eq:4}), we have $T_k ^c (\{ f,g \}) = T_k  (\{ f, g \} ) + \bigo( k^{-1}) |f,g|_3$. By Lemma \ref{lem:estim_norm_cor},
$$[ \Pi_k \tfrac{1}{ik} \nabla^k_X \Pi_k , \Pi_k \tfrac{1}{ik} \nabla^k_Y \Pi_k ] = \bigo(k^{-2}) |f|_2 |g|_2\;.$$ So by Proposition \ref{prop:P2_correction},  it suffices to show that
\begin{gather} \label{eq:9}
 [\Pi_k f \Pi_k , \Pi_k \tfrac{1}{ik} \nabla^k_ Y \Pi_k ] = \bigo (k^{-2}) (|f|_1 |g|_3 + |f|_2 |g|_2 )\;.
\end{gather}
Write $Y = Z+ \con{Z}$ with $Z$ a section of $T^{1,0}M$.  Doing a straightforward computation, we obtain
$$  [ \Pi_k H_k ( f) \Pi_k , \Pi_k \tfrac{1}{ik} \nabla^k_{\con{Z}} \Pi_k ] =  \Pi_k [ H_k(f), \Pi_k ] \tfrac{1}{ik} \nabla^k_{\con{Z}} \Pi_k + \Pi_k [ H_k(f),  \tfrac{1}{ik} \nabla^k_{\con{Z}} \Pi_k ] \Pi_k\;.$$
By Theorem \ref{theo:compi} and Equation (\ref{eq:12}), the first term of the left hand side is a $\bigo( k^{-2}) |f|_2 |Z|_0$. By (\ref{eq:13}), the second term is  a $\bigo( k^{-2})  ( |f|_1 |Z|_1 + | f|_2 | Z|_0 ) $. Consequently,
\begin{gather} \label{eq:5}
 [ \Pi_k H_k ( f) \Pi_k , \Pi_k \tfrac{1}{ik} \nabla^k_{\con{Z}} \Pi_k ] =  \bigo ( k^{-2})  ( |f|_1 |Z|_1 + | f|_2 | Z|_0 )\;.
\end{gather}
Using Lemma \ref{lem:estim_norm_cor} and Equation (\ref{eq:12}), we deduce from (\ref{eq:5}) that
\begin{gather} \label{eq:11}
 [ \Pi_k f \Pi_k , \Pi_k \tfrac{1}{ik} \nabla^k_{\con{Z}} \Pi_k ] =  \bigo ( k^{-2}) ( |f|_1 |Z|_1 + | f|_2 | Z|_0 )\;.
\end{gather}
Taking the adjoint, we get
\begin{gather} \label{eq:10}
  [ \Pi_k f \Pi_k , \Pi_k \tfrac{1}{ik} (\nabla^k_{Z} + \op{div} (Z)) \Pi_k ] =  \bigo ( k^{-2}) ( |f|_1 |Z|_1 + | f|_2 | Z|_0 )\;.
\end{gather}
By Proposition \ref{prop:P1}, $[T_k (f) , T_k ( v) ] = \bigo ( k^{-1}) (|f|_1|v|_1 + |f|_2|v|_0)$  for any function $v \in \Cl^1 (M)$. Applying this to $v = \op{div} (Z)$, we deduce from (\ref{eq:10}) that
\begin{gather} \label{eq:7}
 [ \Pi_k f \Pi_k , \Pi_k \tfrac{1}{ik} \nabla^k_{Z}  \Pi_k ] =  \bigo ( k^{-2})   ( |f|_1 |Z|_1 + | f|_2 | Z|_0 )\;.
\end{gather}
Finally, Equation (\ref{eq:9}) follows from Equations (\ref{eq:11}) and (\ref{eq:7}).
\end{proof}

\section{Proofs of Theorems \ref{theo:compi} and \ref{lem:Tu_symp}} \label{sec:proofs}

We will now use explicitely the structure of the kernel $B_k$ given in Theorem \ref{theo:B_kernel}.

\subsection{The fundamental estimates} \label{sec:fund-estim}

Consider a section $E$ satisfying the same assumptions as in Theorem \ref{theo:B_kernel}.
The existence of such a section is proved in Lemma 3.2 of \cite{oim_symp}.  Let $U$ be the open set where $E$ does not vanish. Let $\varphi \in \Ci (U)$ and $\al_E \in \Om^1(U)$ be defined by
\begin{gather} \label{eq:def_al_E_phi}
 \varphi = - 2 \ln |E|\;, \qquad \nabla^{L \boxtimes \con{L}} E = \tfrac{1}{i} \al_E \otimes E\;.
\end{gather}
Here $\nabla^{L \boxtimes \con{L}}$ is the connection of $L \boxtimes \con{L}$ induced by $\nabla$. So for any vector fields $X$ and $Y$ of $M$, $\nabla^{L \boxtimes \con{L}}_{(X,Y)} = \nabla_X \boxtimes \op{id} + \op{id} \boxtimes \con{\nabla}_Y$, where $(X,Y)$ is the vector field of $M^2$ sending $(p,q)$ into $X(p) \oplus Y(q)$.

By Theorem \ref{theo:B_kernel}, $\varphi$ vanishes along the diagonal $\De$ of $M^2$ and is positive outside $\De$. Furthermore,  $\varphi$ and $\al_E$ satisfy the following properties
\begin{enumerate} [label=(\roman{*}), ref=(\roman{*})]
\item $\al_E$ vanishes on $T_\De (M ^2)$; \label{it:al}
\item  $\varphi$ vanishes to second order along $\De $. For any $x\in M$, the kernel of the Hessian of $\varphi$ at $(x, x)$ is the tangent space to the diagonal; \label{it:phi}
\item \label{it:com}
for any $f \in \Ci(M)$ with Hamiltonian vector field $X$, $g  - \al_E(X,X) $ vanishes to second order along $\De $, where  $g(x,y) = f(x) - f(y)$.
\end{enumerate}
For a proof of these properties, cf. Proposition 2.15, Remark 2.16 and Proposition 2.18 of \cite{oim_symp}.

For any continuous section $\si$ of $A \boxtimes \con{A}$ and $k \in \N$, we let $P_k( \si)$ be the operator acting on $\Cl ^0 ( M , L^k \otimes A)$ with Schwartz kernel $k^n E^k \otimes  \si $. Here and in the sequel, $\mu$ is the Liouville form $\om^n / n!$.

\begin{lemma} \label{lem:estim1}
For any compact subset $K$ of $U$, for any $p \in \N$, there exists $C_{K, p }$ such that for any $\si \in \Cl^0 ( M^2, A \boxtimes \con {A})$ whose support is contained in $K$, we have
$$  \forall k \in \N, \qquad \| P_k (\si)  \|_{op}  \leqslant C_{K, p} | \si |_{K,p}   k^{ - p /2}\;, $$
where $| \si |_{K,p} \in \R_{+}\cup \{ \infty \}$ is the supremum of $|\si (z) | (\varphi (z) )^{-p /2}$ over $K \setminus \De $.
\end{lemma}

\begin{proof}  Assume first that $K$ does not intersect the diagonal of
  $M$. Then $\varphi$ takes positive values on $K$ so there exists $C>0$
  such that $1/C \leqslant \varphi \leqslant C$ on $K$. Consequently
$$|E^k  \otimes  \si | \leqslant | \si |_{K,p} C^{p/2} e^{-k/(2C)}   $$
on $K$ and  we conclude easily.

Assume now that $K \subset V^2$ where $(V, x_i)$ is a coordinate system of $M$ such that $V^2 \subset U$.   By property \ref{it:phi} and the fact that $\varphi$ is positive outside the diagonal, there exists $C>0$ such that
\begin{gather} \label{eq:phi_loc}
 |x-y|^2 / C  \leqslant \varphi (x, y ) \leqslant C |x-y|^2
\end{gather}
on $K$. If the support of $\si$ is contained in $K$, we obtain that  $ | \si (x,y) | \leqslant C^{p /2} | \si |_{K, p} |x - y|^p $ on $V^2$.
Identify $V$ with an open set of $\R^{2n}$. Then we have
\begin{xalignat*}{2}
\int_M \bigl| P_k (\si) ( x, y ) \bigr| \mu (y) & \leqslant k^{n} C^{p /2}
| \si |_{K, p}     \int_V e^{- k |x-y|^2/C} |x - y | ^p \; dy \\ &
\leqslant k^{n} C^{p /2} | \si |_{K, p}    \int_{\R^{2n}} e^{- k |x-y|^2/C}
|x - y | ^p \; dy \\
& =  k^{-p/2}  C^{p /2} | \si |_{K, p}   \int_{\R^{2n}} e^{- |x-y|^2/C} |x
- y | ^p \; dy
\intertext{by doing a convenient change of variable}
& = k^{-p /2} C_1  | \si |_{K, p}
\end{xalignat*}
In the same way we show that
$$ \int_M \bigl| P_k (\si)  ( x, y ) \bigr| \mu (x) \leqslant  k^{-p/2} C_2  |\si |_{K, p} $$
for some $C_2 >0$ independent of $\si$ and $k$.  We conclude by applying Proposition \ref{prop:norme} that
$$ \| P_k (\si) \|_{op} \leqslant C  |\si |_{K,p} k^{ - p /2}  $$
with $C = \max ( C_1, C_2)$.

Consider now any compact subset $K$ of $U$. The diagonal $\Delta$ being
compact, there exists a finite family $(V_i)_{i=1, \ldots , N} $ of open sets of $M$ such
that each $V_i$ is the domain of a coordinate system, $V_i^2 \subset U$ and
$\Delta  \subset \cup V_i^2$. Then $U$ is covered by the $(N+1)$ open sets
$U_0 =  U \setminus \De $, $U_i = V_i^2$, $i=1,\ldots,N$.
Choose a subordinate partition of unity $f_i \in \Ci (U)$, $i =0, \ldots,
N$. If $\si$ is supported in $K$, $f_i \si$ is supported in $K \cap
\op{supp} f_i$ and we have by the first part of the proof for $i=0$ and the  second part for $i =1, \ldots, \ell$,
$$\| P_k ( f_i \si) \|_{op} \leqslant C_i |f_i \si |_{K \cap \op{supp} f_i, p} k^{-p/2} \leqslant C_i |\si |_{K, p} k^{-p/2}$$
for some constants $C_i >0$.
\end{proof}

Recall that we denote by $H_k (f)$ the Kostant Souriau operator (\ref{eq:def_Kostant_souriau}).
\begin{lemma} \label{lem:estim2}
For any $p \in \N$, for any $\si \in \Ci ( M^2, A \boxtimes \con{A})$ supported in $U$ and vanishing to order $p$ along the diagonal, there exists $C>0$ such that for any $ f\in \Cl^2 (M)$, we have
$$  \| P_k (\si)  \|_{op}  \leqslant C  k^{-p/2} , \qquad \bigl\| [H_k(f) , P_k ( \si)] \bigr\|_{op} \leqslant  C k^{-p/2 -1} |f|_2\;.$$
\end{lemma}

\begin{proof}
 It is a consequence of Lemma \ref{lem:estim1}. Set $K= \op{supp} \si$. Using Property \ref{it:phi} as in Equation (\ref{eq:phi_loc}), if follows from Taylor formula that
$ | \si |_{K, p} $ is finite, which proves the first estimate.

To prove the second one, we introduce  $g(x,y) = f(x) - f(y)$ and the vector field $Y= (X,X)$ of $M^2$, where $X$ is the Hamiltonian vector field of $f$.
Then we have on $U$
$$ \bigl[ ( f - i \nabla_X ) \boxtimes \op{id} - \op{id} \boxtimes (f + i \nabla_X ) \bigr]  E = ( g - \al_E (Y)) E\;.$$
Thus
$$ \bigl[ ( f +\tfrac{1}{ik} \nabla^k_X ) \boxtimes \op{id} - \op{id} \boxtimes (f - \tfrac{1}{ik} \nabla^k_X ) \bigr] ( E^k \otimes  \si)
= E^k  \bigl( \bigl( g - \al_E (Y) + \tfrac{1}{ik}  \nabla^{ A \boxtimes \con{A}}_Y \bigr) \si \bigr)
$$
Consequently, using the basic facts on Schwartz kernels recalled in Section \ref{sec:schwartz-kernels}, we obtain
\begin{gather} \label{eq:6}
  [ H_k (f) , P_k ( \si) ] = P_k \bigl( (g - \al_E (Y) ) \si \bigr)  + \tfrac{1}{ik} P_k \bigl(  \nabla^{ A \boxtimes \con{A}}_Y \si \bigr)\;.
\end{gather}
We claim that there exists $C>0$ such that for any $f \in \Cl^2(M)$, we have
\begin{gather} \label{eq:e2}
 |g - \al_E (Y)| \leqslant C \varphi |f|_2
\end{gather}
on $K$. This has the consequence that $ |(g + \al_E (Y) ) \si|_{K, p+2} \leqslant C | f|_2 |\si|_{K,p}$. So by Lemma \ref{lem:estim1}, the first term of the right hand side of (\ref{eq:6}) is a $\bigo( k^{-p/2 -1}) |f|_2$.
To prove Equation (\ref{eq:e2}), introduce a coordinate system $(V,x_i)$ such that the closure of $V$ is a compact subset of $U$. We have on $V^2$
 $$ g(x,y) = \sum_{i=1}^{2n} g_i (x) (y_i - x_i) + \bigo (\varphi ) |f|_2\;,$$
for some functions $g_i \in \Cl^1 ( V)$. Similarly, by Property \ref{it:al}, $\al_E $ vanishes along the diagonal $\De$, thus
$$ \al_E (Y) = \sum_{i=1}^{2n} h_i(x)   (y_i - x_i) + \bigo ( \varphi) |X|_0\;.$$
By Property \ref{it:com}, $( g - \al_E (Y))$ vanishes to second order along $\De$, so  for any $i$,  $g_i (x) = h_i (x) $.  Equation (\ref{eq:e2}) follows.

We will show that there exists $C>0$ such that for any vector field $Z$ of $M^2$ tangent to $\De$,
\begin{gather} \label{eq:8}
 | \nabla^{A \boxtimes \con{A}}_Z \si | \leqslant C \varphi^{p/2} |Z|_1\;.
\end{gather}
By Lemma \ref{lem:estim1}, this has the consequence that the second term of the right hand side of (\ref{eq:6}) is a $\bigo ( k^{-p/2-1}) |f|_2 $, which concludes the proof.

Let us prove (\ref{eq:8}). We denote by $\bigo(N)$ any section vanishing to order $N$ along the diagonal. Observe  that for any vector field $Z$ of $M^2$, $\nabla_{Z}^{A \boxtimes \con{A}} \si$ is a $\bigo( p-1)$. Whenever $Z$ is tangent to $\De$, $\nabla_{Z}^{A \boxtimes \con{A}}$ is a $\bigo( p)$. So if $(V,x_i)$ is a coordinate system as above,
$$ \nabla^{A \boxtimes \con{A}} \si = \sum_{i=1}^{2n} (dy_i - dx_i) \otimes a_i + dx_i \otimes b_i\;,
$$
where $ a_i = \nabla^{A \boxtimes \con{A}}_{\partial_{y_i}} \si$ is a $\bigo( p-1)$ and $b_i = \nabla^{A \boxtimes \con{A}}_{\partial_{x_i} + \partial_{y_i}} \si$ is a $\bigo (p)$.

Now there exists $C'>0$ such that for any vector field $Z$ of $M^2$ tangent to $\De$ of class $\Cl^1$ and supported in $V^2$, we have
$$ | ( d y_i - dx_i ) ( Z) | \leqslant   C' \varphi^{1/2} |Z|_1\;, \qquad | dx_i (Z) | \leqslant C' |Z|_0\;.$$
This proves (\ref{eq:8}) for the vector fields supported in $V^2$. We prove the general case with a partition of unity argument.
\end{proof}

\subsection{The proof}

Recall that we denote by $B_k$ the Schwartz kernel of $\Pi_k$. Let $\psi \in \Ci ( M\times M , \R)$ be equal to $1$ on a neighborhood of the diagonal and supported in $U$. Let $R_k$ be the operator with Schwartz kernel $(1- \psi) B_k$. Introduce the same family $(\si_\ell, \ell \in \Z)$ as in Theorem \ref{theo:B_kernel}. For any $m \in \N$, introduce the operator $R_{m,k}$ so that
\begin{gather} \label{eq:somme}
\Pi_k =  (2\pi)^{-n} \sum_{\ell \in \Z \cap [ -m ,m/2] } k^{-\ell}   P_k ( \psi \si_{\ell} ) + R_k + R_{m,k}\;.
\end{gather}
Each term of the right hand side of  (\ref{eq:somme}) will be denoted generically by $Q_k$. We will prove that when $m$ is sufficiently large, we have
for any vector field $Z \in \Cl^\infty ( M ,T^{1,0}M)$
\begin{gather}
\label{eq:A1}
\bigl[ H_k ( f) , Q_k \bigr] = \bigo( k^{-1}) |f|_2, \\ \label{eq:A2}
\tfrac{i}{k} \nabla^k_{\con{Z}} Q_k = \bigo ( k^{-1}) \\ \label{eq:A3}
\bigl[ H_k ( f) , \tfrac{i}{k} \nabla^k_{\con{Z}} Q_k \bigr] = \bigo ( k^{-2})|f|_2\;.
\end{gather}
After that, we will prove that (\ref{eq:A2}) holds for any $Z$ continuous, and (\ref{eq:A3}) holds for any $Z$ of class $\Cl ^1$. Finally we will make explicit the dependence in $Z$ of the $\bigo$.

\subsubsection*{The principal terms}

For any $\ell \in \Z$, let $Q_{\ell,k} = k^{-\ell} P_k ( \psi \si_{\ell})$. By Lemma \ref{lem:estim2}, we have
\begin{gather}
 [ H_k (f) , Q_{ \ell ,k} ] = \begin{cases} \bigo ( k^{-\ell -1 }) |f|_2 \text{ if } \ell \geqslant 0\;, \\ \bigo ( k^{ \ell/2 -1 } )|f|_2 \text{ if } \ell \leqslant 0\;.
\end{cases}
\end{gather}
This proves that $Q_{ \ell ,k}$ satisfies (\ref{eq:A1}). To prove the remaining estimates, we use that
\begin{gather}\label{eq:18}
\tfrac{i}{k} \nabla^k_{\con{Z}} Q_{\ell ,k} =  P_k (  \al_E(\con{Z},0) \psi \si _{\ell}) + \tfrac{i}{k}  P_k  ((\nabla_{\con{Z}}^A \boxtimes \op{id}) \psi \si_{\ell}  \bigr)\;,
\end{gather}
where $\al_E$ has been introduced in (\ref{eq:def_al_E_phi}). By Theorem \ref{theo:B_kernel}, $\al_E ( \con{Z}, 0)$ vanishes to second order along the diagonal.
By Lemma \ref{lem:estim2}, it comes that
\begin{gather}
\tfrac{i}{k} \nabla^k_{\con{Z}} Q_{ \ell ,k}  = \begin{cases} \bigo ( k^{-\ell -1 })  \text{ if } \ell \geqslant 0\;, \\ \bigo ( k^{ \ell/2 -1 } ) \text{ if } \ell \leqslant 0
\end{cases}
\end{gather}
and
\begin{gather}
 \bigl[ H_k(f), \tfrac{i}{k} \nabla^k_{\con{Z}} Q_{ \ell ,k} \bigr]  = \begin{cases} \bigo ( k^{-\ell -2 }) |f|_2  \text{ if } \ell \geqslant 0\;, \\ \bigo ( k^{ \ell/2 -2 } ) |f|_2 \text{ if } \ell \leqslant 0\;,
\end{cases}
\end{gather}
which prove (\ref{eq:A2}) and (\ref{eq:A3}) for $Q_k = Q_{k,\ell}$.

\subsubsection*{The remainders}

Denote by $B'_k = ( 1- \psi) B_k$ and $B_{m,k}$ the Schwartz kernels of $R_k$ and $R_{m,k}$ respectively. Let $\nabla^k$ be the connection of $A_k \boxtimes \con{A}_k$ induced by the connections of $A$ and $L$.

Recall the class $\bigo_{\infty} ( k^{N})$ introduced after Theorem \ref{theo:B_kernel}. Set $\bigo_{\infty} (k^{-\infty}) := \cap_{N>0} \bigo_{\infty} ( k^{-N})$.

\begin{lemma} \label{lem:remainder}
 $(B_k)$ belongs to $\bigo_{\infty} ( k^{-\infty})$. $(B_{m,k})$ belongs to $\bigo_{\infty} ( k^{n-(m+1)/2})$. In particular, for any smooth vector fields $X_1$, $X_2$ of $M^2$, we have for any $N$
\begin{gather*}
 \nabla_{X_1}^k B'_k = \bigo( k^{-N}), \qquad \nabla_{X_1}^k \nabla_{X_2}^k B'_k = \bigo( k^{-N})\;, \\
 \nabla_{X_1}^k B_{m,k} = \bigo( k^{n +1/2 -m/2} ) , \qquad \nabla_{X_1}^k \nabla_{X_2}^k  B_{m,k} = \bigo( k^{n +3/2 - m/2 } )
\end{gather*}
uniformly on $M^2$.
\end{lemma}

\begin{proof} Observe first  that if $(\Psi( \cdot , k))$ belongs to $ \bigo_{\infty}( k^{N})$ and $\xi \in \Ci ( M)$, then $( \xi \Psi ( \cdot, k ))$ belongs to $\bigo_{\infty} ( k^N)$. The remainder in (\ref{eq:kerprojdevas}) being in  $\bigo_{\infty} ( k^{n-(m+1)/2})$, the same holds for $(B_{m,k})$. By the same reason, $(B_k)$ being in $\bigo_{\infty} (k^n)$, the same holds for $(B'_k)$. Since the pointwise norm of the section $E$ appearing in (\ref{eq:kerprojdevas}) satisfy $|E|<1$ outside the diagonal, $(B_k)$ is a $\bigo ( k^{-\infty})$ on any compact set not intersecting the diagonal. So $(B'_k)$ is in $\bigo ( k^{-\infty})$. This actually implies that $(B'_k)$ is in $\bigo_{\infty} ( k^{-\infty})$. Indeed, $\bigo_{\infty} ( k^{n} ) \cap \bigo ( k^{-\infty}) = \bigo_{\infty} ( k^{-\infty}) $, which follows from  the basic interpolation formula:
 for any open set $V$ of $\R^{m}$ and compact subset $K$ of $V$, there exists $C>0$ such that for any smooth function $f$ on $V$, we have
$$ \sum_{|\al|=1} \op{sup} _K  |\partial ^\al f| \leqslant C \Bigl( \op{Sup}_V |f| \Bigr)^{1/2}   \Bigl( \op{sup}_V |f| +\sum_{|\al|=2}  \op{sup}_V  |\partial ^\al f|   \Bigr)^{1/2}\;.$$
A proof may be found \cite{Sh}, Lemma 3.2.
\end{proof}


By writing the Schwartz kernels of $[H_k (f) , Q_k]$, $\nabla^k_{\con{Z}}Q_k$ and $[H_k (f) ,\nabla^k_{\con{Z}}Q_k]$ in terms of the Schwartz kernel of $Q_k$, we deduce from Lemma \ref{lem:remainder}  that when $m$ is sufficiently large, $R_k$ and $R_{m,k}$ satisfy (\ref{eq:A1}), (\ref{eq:A2}), (\ref{eq:A3})  for smooth $f$; so far, without specifying the dependence of the $\bigo$'s on $f$.

To make explicit this dependence, we use the following fact. Consider any family $(\tau_k \in \Ci ( M^2, A_k \boxtimes \con{A}_k) , \; k \in \N)$ and assume that there exists $N \in \R$ such that for any smooth vector field $X$ of $M^2$, we have
$$ \nabla^k_X \tau_k = \bigo ( k^{-N}) $$
uniformly on $M^2$. Then there exists $C>0$, such that for any continuous vector field $X$,
\begin{gather} \label{eq:19}
 | \nabla^k_X \tau_k | \leqslant C k^{-N} |X|_0
\end{gather}
on $M^2$. To prove that, we write $X$ in local smooth frames and use that $\nabla^k_{X_1+ X_2} = \nabla^k_{X_1} + \nabla^k_{X_2}$ and $\nabla^k_{gX} = g \nabla^k_X$. This proves (\ref{eq:A1}) and (\ref{eq:A3})  for $Q_k= R_k$ or $ R_{m,k}$ with actually $|f|_1$ instead of $|f|_2$, since $|X_0|$ only depends on the first derivatives of $f$.

\subsubsection*{Dependence in $Z$}

We claim that if $(Q_k, k \in \N)$ satisfies (\ref{eq:A2}) for any  $Z \in \Cl^\infty (M, T^{1,0}M)$, then for any $Z \in \Cl^{0} ( M , T^{1,0} M)$, we have
$$ \tfrac{i}{k} \nabla^k_{\con{Z}} Q_k = \bigo ( k^{-1})|Z|_0\;.
$$
The proof is the same as the one to get (\ref{eq:19}): write $Z$ in local smooth frames of $T^{1,0} M$.

Similarly, assume that (\ref{eq:A2}) and (\ref{eq:A3}) hold for any smooth section of $T^{1,0}M$, then for any $Z \in \Cl^1 (M, T^{1,0} M)$, we have
$$ \bigl[ H_k ( f) , \tfrac{i}{k} \nabla^k_{\con{Z}} Q_k \bigr] = \bigo ( k^{-2})(|f|_2 |Z|_0 + |f|_1 |Z|_1)\;.$$
The proof is the same by using now that
$$  \bigl[ H_k ( f) , \tfrac{i}{k} \nabla^k_{g \con{Z}} Q_k \bigr] = g \bigl[ H_k ( f) , \tfrac{i}{k} \nabla^k_{\con{Z}} Q_k \bigr] + \tfrac{i}{k} (X.g) \tfrac{i}{k} \nabla^k_{\con{Z}}Q_k\;, $$
where $X$ is the Hamiltonian vector field of $f$.

\section{Bargmann space}

In this section, we prove Theorem \ref{thm:toepl-barg-comp} and the version (\ref{eq:P2_bargmann}) of (P2). It is sufficient to prove these estimates for $\hbar =1$. Indeed, recall that for any $\hbar >0$ and $f \in L^{\infty}( \C^n)$, we denote  by $\B_{\hbar}$ the Bargmann space and by $ T_{\hbar} (f)$ the Toeplitz operator $\Pi_{\hbar} f : \B_{\hbar} \rightarrow \B_{\hbar}$.
Then we have a Hilbert space isomorphism
$$ U_\hbar: \B_{\hbar} \rightarrow \B_1, \qquad \xi \rightarrow \hbar^{n/2} \xi (\hbar^{1/2} \cdot).$$
We easily check  that
\begin{equation}
\label{eq-qcomresc}
T_{\hb} (f) = U_{\hbar} ^* T_1 ( f_{\hbar}) U_{\hbar}\;,
\end{equation}
where $f_{\hbar} (x) = f ( \hbar ^{1/2} x) $. Recall the semi-norm $| \cdot |_k'$ introduced in (\ref{eq:semi-norm-prime}). Using that $|f_{\hbar}|'_k = \hbar^{k/2} |f|'_k$, we see that Equation (\ref{eq:P2_bargmann}) and Theorem \ref{thm:toepl-barg-comp} with  $\hbar = 1$ imply the same results for any $\hbar$.

Instead of $\B_{1}$, it will be more convenient to work with the closed subspace $\B$ of $L^2 ( \C^n, \mu)$ consisting of the functions $\xi$ satisfying $ \partial \xi /\partial \con{z}_i = - \tfrac{1}{2} z_i \xi$ for  $i = 1, \ldots , n$.
 $\B _{1}$ and $\B$ are isomorphic Hilbert space through the unitary map $\xi \rightarrow \xi e^{-|z|^2/2}$. Furthermore, for any $f \in  L^{\infty} ( \C^n)$, this unitary map conjugates $T_{1} (f)$ with $T (f) := \Pi f: \B \rightarrow \B$, where $\Pi$ is the orthogonal projector of $ L^2 ( \C^n, \mu)$ onto $\B$. So our goal is to prove the following.

\begin{thm} \label{thm:toepl-barg-comp-sans-hbar}
For any $N \in \N$, there exists $C_N >0$ such that for any $f \in \Cb ^{2N}( \C^n)$ and $g \in \Cb ^{N} ( \C^n)$, we have
$$ T (f) T(g) = \sum_{\ell = 0 } ^{N-1} (-1)^{\ell} \hb^\ell \sum_{\al \in \N^n,\; |\al |= \ell} \frac{1}{\al !} T \bigl( ( \partial_{z}^{\al} f)( \partial_{\con{z}}^{\al} g) \bigr) + R_N( f,g)\;, $$
where $\|R_{N} (f,g) \|_{op} \leqslant C_N \sum_{m=0}^{N} |f|'_{N+m} |g|'_{N - m}$.
\end{thm}

It is well-known that the Schwartz kernel of $\Pi$ is given by
\begin{equation}
\label{eq-schkerbarg}
\Pi ( u, v) = ( 2 \pi ) ^{-n} e^{ -\frac{1}{2} ( |u|^2 + |v|^2 ) + u \cdot \con{v}} , \qquad u,v \in \C^n\;.
\end{equation}
It satisfies the following two identities
\begin{gather} \label{eq:identity_bargmann_norm}
 |\Pi(u,v)| = ( 2 \pi ) ^{-n} e^{ -\frac{1}{2} |u -v|^2}\;, \\ \label{eq:identity_bargmann_composition}
\Pi(u,v) \Pi( v,w) = ( 2 \pi ) ^{-n} e^{- (v - u ) \cdot ( \con{v} - \con {w} )} \Pi( u , w)\;.
\end{gather}
Let $W : \C^{4n} \rightarrow \R$ be the weight given by
$$W =  1 + |z_1 -z_2| + |z_2 - z_3 | + |z_3 - z_4| .$$
Let $N \in \N$.  For any measurable function $g : \C^{4n} \rightarrow \C$ such that $ |g| W^{-N}$ is bounded, introduce the function of $\C^{2n}$:
$$ K( g) ( x_1 , x_4) = \int_{\C^{2n} } \Pi (x_1,x_2 ) \Pi ( x_2, x_3)  \Pi (x_3 , x_4)  g (x_1, x_2, x_3 , x_4 )\; \mu (x_2) \mu ( x_3)\;. $$

\begin{lemma} \label{lemma:bound}
$K(g)$ is the Schwartz kernel of a bounded operator of $L^ 2 ( \C^n , \mu)$. Its uniform norm satisfies
$$ \| K(g) \| \leqslant C_N \sup_{ x\in \C^{4n} } \bigl( |g(x) | W^{-N}(x)\bigr)   $$
for some constant $C_N$ independent of $g$.
\end{lemma}

\begin{proof}
This follows from Schur test. Indeed, by (\ref{eq:identity_bargmann_norm}), we have for any $x_1 \in \C^n$
\begin{xalignat*}{2}  & \int_{\C^n} | K(g) (x_1, x_4) |\mu (x_4)  \\
& \leqslant
 \int_{\C^{3n} }  | \Pi (x_1,x_2 ) \Pi ( x_2, x_3) \Pi (x_3 , x_4)  g (x_1, x_2, x_3 , x_4 ) | \; \mu (x_2) \mu ( x_3) \mu (x_4) \\  & \leqslant   C_N \sup_{ x\in \C^{4n} } \bigl( |g(x) | W^{-N}(x)\bigr)\;,
\end{xalignat*}
where $C_N$ is the constant
\begin{xalignat*}{2}
 C_N =  &  \int_{\C^{3n}} e^{ -\frac{1}{2} ( |x_1 -x_2|^2 + |x_2 - x_3|^2 + |x_3- x_4|^2)} W( x_1,x_2,x_3,x_4) ^N  \; \mu (x_2) \mu ( x_3) \mu (x_4) \\
= &  \int_{\C^{3n}} e^{ -\frac{1}{2} ( |y_1|^2 + |y_2|^2 + |y_3|^2)} (1 + |y_1|+ |y_2| + |y_3| )^N  \; \mu ( y_1) \mu (y_2) \mu ( y_3)\;.
\end{xalignat*}
Similarly,  $\int_{\C^n} | K(g) (x_1, x_4) |\mu (x_1) \leqslant  C_N \sup_{ x\in \C^{4n} } \bigl( |g(x) | W^{-N}(x)\bigr) $.
\end{proof}
Observe that for $f_1$, $f_2 \in \Cb^0 ( \C^n)$, $K ( 1 \boxtimes f_1 \boxtimes f_2 \boxtimes 1)$ is the product of Toeplitz operators $T(f_1) T(f_2)$. In particular $K(  1 \boxtimes 1\boxtimes f_2 \boxtimes 1) = T(f_2)$. Here, we denote by $ 1 \boxtimes f_1 \boxtimes f_2 \boxtimes 1$ the functions sending $(x_1,x_2,x_3,x_4)$ to $f_1(x_2)  f_2(x_3)$. In the sequel, abusing notations, we will sometimes use the notation $K( f_1 (z_2) f_2 (z_3))$ instead of $K (1 \boxtimes f_1 \boxtimes f_2 \boxtimes 1)  $.

\begin{lemma} \label{lem:integration_part}
Let $g \in \Cl^1 ( \C^{4n})$ be such that $|g| W^{-N} $ and $ |f| W^{-N} $ are bounded  where $f$ is any partial derivative of $g$. Then we have
$$ K \bigl( (\con{z}_{2,i} - \con{z}_{3,i} ) g \bigr) = K \bigl( \partial g / \partial z_{2,i} \bigr)\;,  \qquad  K \bigl( ({z}_{3,i} - {z}_{2,i} ) g \bigr) = K \bigl( \partial g / \partial \con{z}_{3,i} \bigr)\;. $$
\end{lemma}
\begin{proof}
By (\ref{eq:identity_bargmann_composition}), we have
\begin{gather*}
\Bigl( \frac{\partial}{ \partial  z_{2,i}} + (\con{z}_{2,i} - \con{z}_{3,i} ) \Bigr) \bigl( \Pi ( z_1, z_2) \Pi(z_2, z_3) \bigr)  = 0\;, \\
\Bigl( \frac{\partial}{ \partial  \con{z}_{3,i}} +  (z_{3,i} - z_{2,i} ) \Bigr)  \bigl( \Pi ( z_2, z_3) \Pi(z_3, z_4) \bigr) = 0\;.
\end{gather*}
And the result follows by integrating by part.
\end{proof}

Consider now $N \in \N$ and $f,g \in \Cb^{2N} ( \C^n)$. We will compute $K( f(z_2) g(z_3))$ by replacing $f(z_2)$ by
its Taylor expansion around $z_3$
$$ f(z_2) = \sum_{ \al, \be \in \N^n, \; |\al| + |\be| < 2N} \frac{1}{\al ! \be !} f_{\al, \be} ( z_3) (z_2 - z_3)^{\al} ( \con{z}_2 - \con{z}_3)^\be + r_N (z_2, z_3)\;.$$
Here for any $\al$, $\be$, we denote by $f_{\al, \be}$ the derivative $\partial_{z}^{\al} \partial_{\con{z}}^{\be} f $. Furthermore, for any $z_2, z_3 \in \C^n$,
$$ |r_N (z_2, z_3) | \leqslant C_{N}' |f|'_{2N} \bigl( 1 + |z_2 - z_3 |\bigr)^{2N}$$
for some constant $C'_N$ independent of $f$. By Lemma \ref{lemma:bound},
$$ \|  K(r_N(z_2, z_3) g(z_3) ) \| \leqslant C_{2N} C'_N |f|'_{2N} | g|_0'\;.$$
Denote by $P_{\al ,\be}$ the operator $K \bigl( f_{\al, \be} ( z_3) (z_2 - z_3)^{\al} ( \con{z}_2 - \con{z}_3)^\be g(z_3) \bigr)$. We have
$$ T(f) T(g) = \sum _{ \al, \be \in \N^n, \; |\al| + |\be| < 2N} \frac{P_{\al, \be}}{\al ! \be !} + K ( (r_N(z_2, z_3) g(z_3) )\;.$$
In the sequel we say that $\be \leqslant \al$ if for any $i=1,\ldots , n$, we have $\be (i) \leqslant \al (i)$.

\begin{lemma}
Let $\al$, $\be \in \N^n$ be such that $|\al| + | \be| < 2N$. Then if $|\al| <N$ and $\be \leqslant \al$, we have
\begin{gather} \label{eq:Palbe}
 P_{\al , \be}= \frac{ \al !  ( -1)^{|\al - \be|} } { ( \al - \be ) !} T ( \partial_{\con{z}}^{\al - \be} (f_{\al, \be } g))
\end{gather}
and otherwise
\begin{gather} \label{eq:Palbe_neg}
 \| P_{\al,\be} \|   \leqslant C_N \sum_{m= 0}^N | f|'_{N+m} |g|'_{N-m}
\end{gather}
for some constant $C_N$ independent of $f$ and $g$.
\end{lemma}

\begin{proof}
First, if for some $i$ we have $\be (i) > \al (i)$, then by the first identity of Lemma \ref{lem:integration_part}, we have $P_{\al ,\be}=0$ and (\ref{eq:Palbe_neg}) is satisfied. Assume now that $\be \leqslant \al$.  By the first identity of Lemma \ref{lem:integration_part},
$$ P_{\al, \be} = \frac{ \al ! }{(\al - \be) !} K   \bigl( f_{\al, \be} ( z_3) (z_2 - z_3)^{\al - \be }  g(z_3) \bigr)\;. $$
If $| \al | \geqslant N$, then $| \al | + |\be| + |\al - \be| = 2 |\al| \geqslant 2N$, so we can find  a multi-index $\ga \in \N^n$ such that $| \al | + |\be | + |\ga| = 2N$ and $\ga \leqslant \al - \be$. By the second identity of Lemma \ref{lem:integration_part}, we have
$$  P_{\al, \be} = \frac{ \al ! (-1)^{ |\ga|}}{(\al - \be) !} K   \bigl(  (z_2 - z_3)^{\al - \be - \ga }\bigl( \partial_{\con{z}}^\ga ( f_{\al ,\be} g) \bigr) (z_3)    \bigr)\;. $$
Then expanding $ \partial_{\con{z}}^\ga ( f_{\al ,\be} g)$ and applying Lemma \ref{lemma:bound} with the weight $W^{-2N}$, we deduce that (\ref{eq:Palbe_neg}) is satisfied. Finally, assume that $|\al| \leqslant N$ and $\be \leqslant \al$, so that $f_{\al , \be}g$ is of class $\Cl^{ |\al - \be|}$. Then by  second identity of Lemma \ref{lem:integration_part}, we have
$$  P_{\al , \be} = \frac{ \al ! ( - 1)^{|\al - \be| }}{(\al - \be) !} K   \bigl(  \bigl( \partial_{\con{z}}^{\al - \be} ( f_{\al , \be} g) \bigr) (z_3)    \bigr)\;, $$
and we deduce (\ref{eq:Palbe}).
\end{proof}

To finish the proof of Theorem \ref{thm:toepl-barg-comp-sans-hbar}, we have
the following algebraic identity based on the fact that $ \partial_{\con{z}} ^\gamma f_{\al, \be} = f_{\al , \be + \ga}$.
\begin{lemma}
For any $\al \in \N^n$, we have
$$ \sum_{\be \in \N^n, \; \be \leqslant \al } \frac{  ( - 1)^{|\al - \be| }}{\be ! (\al - \be) !}  \partial_{\con{z}}^{\al - \be} (f_{\al, \be } g) = \frac{1}{\al !} f_{\al,  0} g_{0 , \al}\;. $$
\end{lemma}
\begin{proof} Setting $\ga = \al - \be$, we have
\begin{gather*}
  \sum_{\be \in \N^n, \; \be \leqslant \al } \frac{  ( - 1)^{|\al - \be| }}{\be ! (\al - \be) !}  \partial_{\con{z}}^{\al - \be} (f_{\al, \be } g) = \sum_{\ga \in \N^n, \; \ga \leqslant \al } \frac{  ( - 1)^{|\ga| }}{(\al - \ga ) !  \ga !}  \partial_{\con{z}}^{\ga} (f_{\al,  \al - \ga } g) \\
= \sum_{\delta, \ga \in \N^n, \; \delta \leqslant \ga \leqslant \al } \frac{  ( - 1)^{|\ga| }}{(\al - \ga ) !  \delta! ( \ga - \delta)!}  f_{\al, \al - \delta} g_{0,\delta } = \sum_{\delta \in \N^n, \; \delta \leqslant \al } \frac{1}{ \delta!}  C(\delta, \al)f_{\al, \al - \delta} g_{0,\delta }\;,
\end{gather*}
where
\begin{gather*}
 C( \delta , \al ) =\sum_{\ga \in \N^n, \; \delta \leqslant \ga \leqslant \al }  \frac{  ( - 1)^{|\ga| }}{(\al - \ga ) ! ( \ga - \delta)!} \\
= \sum_{\la \in \N^n, \; \la \leqslant \al - \delta } \frac{ (-1)^{ |\la - \delta|}}{ (\al - \delta - \la)! \la!} = \begin{cases} (-1)^{|\al| } \text{ if } \al = \delta\;, \\ 0 \text{ otherwise}\;,
\end{cases}
\end{gather*}
which concludes the proof.
\end{proof}

Let us explain now the proof of estimate (\ref{eq:P2_bargmann}). As for the proof of Theorem \ref{thm:toepl-barg-comp}, we can assume that $\hbar =1$ and work in $\B$ instead of $\B_1$. So we have to prove the following.
\begin{thm} \label{theo:P2-bargmann-space-C3}
There exists $\beta ''$ such that for any $f, g \in \Cb^3 ( \C^n)$, we have
\begin{gather*}
 \Bigl\|  \bigl[ T(f),T(g) \bigr] - i T (\{f,g\}) \Bigr\|_{op} \leq \beta''   \bigl( |f|_1' |g|_3' + |f|_2'|g|_2' + |f|_3'|g|'_1 \bigr)\;.
\end{gather*}
\end{thm}

Consider the trivial holomorphic line bundle $L$ over $\C^n$ with canonical frame $s$. Define the metric of $L$ so that $|s|^2 (z)  = e^{ - |z |^2}$. Then the space $\Hilb$ of holomorphic sections of $L$ with finite $L^2$-norm is isomorphic with $\B_1$ by the map sending a function $\xi \in \B_1$ to the section $\xi s \in \Hilb$.
If we work with the unitary frame $t= e^{-|z|^2/2} s$ instead of $s$, we get an isomorphism between $\B$ and $\Hilb$ by sending $\xi$ to  $\xi t$. In the sequel we identify in this way $\Hilb$ with $\mathcal{B}$ and more generally the space of continuous sections of $L$ with $\Cl^0 ( \C^n)$.

A straightforward computation shows that the Chern connection $\nabla$ of $L$ is given in terms of $t$ by
$$ \nabla t = \al \otimes t \qquad \text{with} \qquad \al = \frac{i}{2} \sum( z_i d\con{z}_i - \con{z}_i d z_i)\;. $$
The curvature is $\frac{1}{i} \om$ with $\om$ the symplectic form $ \om = i \sum dz_i \wedge d \con{z}_i$.
For any function $f \in \Cl^1 ( \C^n)$, introduce the Kostant-Souriau operator $$H(f) = f - i \nabla_X\;,$$ where $X$ is the Hamiltonian vector field of $f$. It acts on functions from $\Ci ( \C^n)$ by $H(f) = f - i X - i \al (X) $.
One easily checks the Kostant-Souriau and Tuynman formulas
\begin{gather*}
[ H(f), H(g) ] = i H ( \{ f, g \} )\;, \qquad f,g  \in \Cb^1 ( \C^n)\;, \\
 \Pi H(f) \Pi = T( f - \Lap f)\;, \qquad f \in \Cb^2( \C^n)\;,
\end{gather*}
where $\Lap = - \sum \partial^2 / \partial z_i \partial \con{z}_i$.
Furthermore we have the following result similar to Theorem \ref{theo:compi}.
\begin{lemma}
There exists $C>0$ such that for any $f \in \Cb^2 (\C)$, one has
$$\bigl\| [ H(f) , \Pi ] \bigr\|_{op} \leqslant C | f_2|'\;.$$
\end{lemma}
\begin{proof}
By a straightforward computation, one checks first that the Hamiltonian vector field $X$ of $f$ is given by $ X = i \sum \bigl(   (\partial f/ \partial \con{z}_i) \partial_{z_i}   - (\partial f/ \partial z_i) \partial_{\con{z}_i} \bigr) $. Then the Schwartz kernel of $\nabla_X \circ \Pi$ is $ i \sum (\con{v}_i - \con{u}_i )   (\partial f/ \partial \con{z}_i) (u) \Pi (u,v)$. So the Schwartz kernel of the commutator $[ H(f) , \Pi ]$ is $m(u,v) \Pi (u,v)$ with
$$ m(u,v) = f(u) - f(v) - \sum (u_i -v_i) \frac{\partial f }{\partial z_i} (v)
- \sum (\con{u}_i -\con{v}_i) \frac{\partial f}{\partial \con{z}_i }(u)\;. $$
Replacing $f(u)$ by its Taylor expansion at $v$, we get that
$$ | m(u , v)  | \leqslant C |f|_2' ( 1 + |u - v|^2)$$
for some constant $C$ independent of $f$. Applying Schur test as in the proof of Lemma \ref{lemma:bound}, we conclude the proof.
\end{proof}

Now the proof of Theorem \ref{theo:P2-bargmann-space-C3} is completely similar to the one of Proposition \ref{prop:P2}, where instead of Proposition \ref{prop:P1_kahler} one uses directly Theorem \ref{thm:toepl-barg-comp-sans-hbar} with $N=1$.


\medskip
\noindent{\bf Acknowledgments.} We thank Yohann Le Floch, Nicolas Lerner, St\'{e}phane Nonnenmacher,
 Johannes Sj\"{o}strand, San V\~{u} Ng\d{o}c and Steve Zelditch for useful discussions.

\bigskip

\noindent
\begin{tabular}{ll}
Laurent Charles & Leonid Polterovich \\
 UMR 7586, Institut de Math\'{e}matiques  & Faculty of Exact Sciences \\
de Jussieu-Paris Rive Gauche &  School of Mathematical Sciences\\
Sorbonne Universit\'{e}s, UPMC Univ Paris 06 & Tel Aviv University \\
F-75005, Paris, France & 69978 Tel Aviv, Israel \\
laurent.charles@imj-prg.fr & polterov@post.tau.ac.il\\
\end{tabular}

\end{document}